\documentclass[journal]{IEEEtran}
\ifx\pdfoutput\undefined
\bibliography{mybib}
\usepackage{graphicx}
\else
\usepackage[pdftex]{graphicx}
\fi
\usepackage{url}
\usepackage{array}
\usepackage{stfloats}
\usepackage{amssymb}
\usepackage{amsmath}
\usepackage{amsthm}
\usepackage{cite}
\usepackage{enumerate}
\usepackage{epstopdf}
\usepackage{color}
\DeclareGraphicsExtensions{.pdf,.png,.jpg}
\begin{document}
 \title{Spatial DCT-Based Channel Estimation in Multi-Antenna Multi-Cell
  Interference Channels}
 \author{
 \IEEEauthorblockN{ Maha~Alodeh,~\IEEEmembership{Student~Member,~IEEE}, Symeon Chatzinotas, \IEEEmembership{Senior~Member,~IEEE,}~\newline
 Bj\"{o}rn Ottersten, \IEEEmembership{Fellow~Member,~IEEE}\thanks{Copyright (c) 2014 IEEE. Personal use of this material is permitted. However, permission to use this material for any other purposes must be obtained from the IEEE by sending a request to pubs-permissions@ieee.org.\newline
Maha Alodeh, Symeon Chantzinotas and  Bj\"{o}rn Ottersten are
with Interdisciplinary Centre for Security Reliability and Trust (SnT) at the University
of Luxembourg, Luxembourg. E-mails:\{ maha.alodeh@uni.lu, symeon.chatzinotas
@uni.lu and bjorn.ottersten@uni.lu\}. B. Ottersten is also with Royal Institute of Technology (KTH),
Stockholm, Sweden. E-mail\{bjorn.ottersten@ee.kth.se\}.\newline
This work is supported by Fond National de la Recherche Luxembourg (FNR)
projects,
project ID:4919957.}}\\} 
\maketitle
\IEEEpeerreviewmaketitle
\begin{abstract}
This work addresses channel estimation in multiple antenna multicell interference-limited networks. Channel state information (CSI) acquisition is vital for interference mitigation. Wireless networks often suffer from multicell
interference, which can be mitigated by deploying beamforming to spatially direct
the transmissions.  The accuracy of the estimated CSI plays
an important role in designing accurate beamformers that can control
the amount of interference created from simultaneous spatial transmissions
to mobile users. Therefore, a new technique based on the structure of the spatial covariance
 matrix and the discrete cosine transform (DCT) is proposed to enhance channel estimation in the presence of interference. Bayesian estimation and
Least Squares estimation frameworks are introduced
by utilizing the DCT to separate the overlapping spatial paths that
create the interference. The spatial domain is thus exploited to mitigate the contamination which is able to discriminate
 across interfering users. Gains over conventional channel estimation techniques are presented in our simulations which are also valid for a small number of antennas.\\
\end{abstract}

 \section{Introduction}
 Interference is the most critical factor for designing and scaling wireless networks, as it leads to the spectrum scarcity-congestion problem. Moreover, the design paradigm for cellular networks has been shifted from partial frequency reuse to full frequency reuse enhancing the spectrum utilization, and thus
making the problem of interference more acute. Therefore, the design of
future networks will require collaborating base stations to jointly serve their
users or to smartly mitigate the interference. This can be enabled by exchanging data information and channel state information (CSI). These designs have
received much attention in the literature, but their main drawback is the requirement of backhaul exchange for users' data, which requires major
upgrades of current infrastructure especially when CSI changes rapidly \cite{zakhour}-\cite{jun}. In an effort to tackle the interference issue, without data sharing over the backhaul network, various coordination techniques
has been proposed to handle the limited data exchange scenarios; for example the authors in \cite{zakhour} exploit the availability of CSI at base stations to design
precoding techniques that minimize the interference created by BS transmissions
by maximizing the signal to leakage noise ratio (i.e. virtual signal to interference
noise ratio), while the work in \cite{jun} investigates zero forcing beamforming in a coordinated multicell environment. Deploying these techniques requires accurate channel state information to design the suitable transmit and receive beamforming. 

The CSI acquisition techniques can be categorized into feedback and reciprocity
techniques. In the feedback systems, a training sequence is broadcasted by
the BS which is measured by users  and
a limited feedback link is considered from the users to the base station. In \cite{njindal}-\cite{adaptive}
such a mode is considered. In \cite{njindal}, the authors have proved that in order to
achieve full multiplexing gain in the MIMO downlink channel  in the high signal to noise ratio (SNR) regime,
the required feedback rate per user grows linearly with the SNR (in dB). The main result in \cite{yoo} is that the extent of CSI feedback
can be reduced by exploiting multi-user diversity. While in \cite{grassman}, it is shown that non-random vector
quantizers can significantly increase the MIMO downlink throughput. In \cite{yoo}, the authors design
a joint CSI quantization, beamforming and scheduling algorithm to attain optimal throughput
scaling. Authors in \cite{adaptive} present an investigation of  how many feedback bits per user are necessary to maintain the optimal multiplexing gain in a K-cell MIMO
interference channel. In the second mode, the distinguishing feature of such systems is the concept of reciprocity, where the uplink channel is utilized as an estimate of the downlink channel
reducing the feedback requirements. This is one of the main advantages of a TDD architecture in low mobility scenarios\cite{per}, as it eliminates the need for feedback, and joint uplink training combined with the reciprocity of the wireless medium are sufficient to
estimate the
desired CSI. TDD has the advantage over direct feedback since the users'
terminal do not need to estimate their own channel. In TDD systems, the base station estimates the channel state information (CSI) based on uplink training
sequences over the same frequency band, and then uses it to generate the beamforming vector in the downlink transmissions\cite{wolfgang}-\cite{gershman}.\smallskip

CSI is typically acquired by exploiting finite-length training sequences in the presence of inter-cell interference. Therefore, in a synchronous setting,
the training sequences from neighboring cells would contaminate
each other. While in an asynchronous setting, the training sequences are contaminated by
the downlink data transmissions. Recently, the problem of non-orthogonality of training sequences
has been thoroughly investigated \cite{jose}-\cite{gan}. It is pointed out in \cite{jose} that training sequence contamination
presents a huge challenge for performance and a robust precoding technique
is proposed to handle this kind of interference. Specifically,
it is shown that training sequence contamination effects \cite{larsson}-\cite{Jindal},
\cite{gan}-\cite{nossekt}
(i.e., the reuse of non-orthogonal training sequences across
interfering cells) causes the interference rejection performance
to quickly saturate with the number of antennas, thereby
undermining the value of MIMO systems in cellular networks\cite{larsson}.

 The problem of training sequence contamination can occur in two different scenarios:
 \begin{itemize}
 \item Inter-cell training sequence contamination: presents the most common scenario in which the users in the same cell utilize orthogonal training sequences while this set is reused in other cells.\smallskip
 \item Intra-cell training sequence contamination: most research ignores this
 kind of contamination as they assume no training sequence reuse
 within the same cell and the training sequences are orthogonal. However, it is
 worth investigating the scenario in which the number of simultaneously tackled transmissions in single cell is greater than number of training sequences\cite{maha_IWCMC}.\smallskip
 \end{itemize}
Although
we mention two types of contamination, handling the contamination in both scenarios is
the same. We focus on the first type of contamination in this work.
Various contamination mitigation techniques have been investigated in the literature;
 a novel asynchronous training sequence transmission is proposed in \cite{tdd}, where different time slots are allocated to the users who utilize the same
training sequence in different cells. The channel estimates are decontaminated from the downlink
interference by adding additional antennas to estimate and cancel the interference.
On the other hand, the effect of pilot contamination is thoroughly studied in the literature
and been characterized in \cite{larsson}-\cite{Jindal} using large system
analysis for different channels. Training sequence allocation techniques are proposed in \cite{nossekt}-\cite{haifan},
to find the optimal set of users that simultaneously utilizes the same training
sequence. The allocation schemes vary according to the considered scenario.
The result in 
\cite{nossekt} depends on the channel power and strongest interferer while
the work in \cite{haifan} relies on second order statistics to pick the users who have the lowest overlap in their subspaces.\smallskip 

In this paper, we tackle the problem of training sequence contamination in correlated
multiple user single input multiple output (MU-SIMO) in the uplink of TDD scenarios.
 Although
the multiple antennas at users' terminal provide additional degrees of freedom,
they present a source of additional complexity from practical perspective which is preferable to be tackled at the BS. Moreover, the assumption
of exploiting single antenna at the users' terminal can be motivated by the fact
that multiplexing multiple streams per users is less beneficial than sending
single stream exploiting receive combining techniques
 \cite{receive_emil}. Based on this observation,
we assume a uniform linear array (ULA) at the base station and exploit the embedded
information in second order statistics about users' favorable and unfavorable directions.

From a different perspective, orthogonal transformations are exploited
in the literature to virtually represent the information in a different domain which simplifies the analysis. In \cite{akbar}, a channel modelling problem for a single user multiple input multiple output (SU-MIMO) is investigated using orthogonal transformations to provide a geometric interpretation of
the scattering environment. This virtual transformation reveals
two important aspects:
the number of parallel channels and the
level of diversity and clarify their impact on capacity calculations.
In
this paper, we exploit the orthogonal transformations' capability of virtually
converting the information into a different domain and based on this we propose several signal
processing algorithms. One of the orthogonal transformations that has appealing
characteristics is the discrete cosine transform (DCT). DCT is utilized to compress the
dimensions used for channel estimation, thus facilitating
interference mitigation at the estimation process.
 The contributions of this paper can
be summarized as\smallskip

\begin{itemize}

\item Improving the contamination mitigation performance of the traditional Bayesian estimation (BE) at the estimation step. The authors in
\cite{haifan} exploit BE to mitigate the contamination and they analyzed its
performance using large system analysis. In this work, we combine BE with
DCT in different algorithms to decontaminate the training sequences at the
estimation process. Moreover, a modification is employed on traditional BE to further improve the contamination rejection from the target estimate. Numerically, the proposed algorithms are shown to outperform the estimation
algorithm
proposed in \cite{haifan}. 

\item Traditionally, least squares estimation (LS) is not capable of discriminating the interference. We apply LS in the DCT domain to mitigate the contamination at the estimation
process
by exploiting DCT characteristics.\smallskip 
\item Training sequence
allocation techniques are proposed to find the optimal set of users that utilize
the same training sequences.  In an effort to utilize the multiuser
diversity concept as in \cite{nossekt}-\cite{haifan}, we propose joint training sequence allocation and
DCT compression to combine the benefit of both schemes. These allocation techniques further suit the nature of enhanced LS and BE thereby outperforming
the traditional techniques.

\end{itemize}
\vspace{-0.4cm}
\subsection{Notation}
The adopted notations in the paper are as follows: we
use uppercase and lowercase boldface to denote matrices and vectors. Specifically, $\mathbf{I}_K$
denotes the $K\times K$ identity matrix. Let $\mathbf{X}^T$, $\mathbf{X}^*$ and
$\mathbf{X}^H$ denote the transpose, conjugate, and conjugate transpose
of a matrix $\mathbf{X}$ respectively. $\mathbb{E}$ refers to the expectation,
$\|\cdot\|$
denotes the Frobenius norm, and $\Vert\cdot\Vert_0$ denotes the zero norm. The Kronecker product of two matrices $\mathbf{X}$
and $\mathbf{Y}$ is denoted by $\mathbf{X}\otimes\mathbf{Y}$. The notation
used  for $\triangleq$ is used for definitions. $\mathbf{1}^{K\times K}$, $\mathbf{0}^{K\times K}$
are the matrices of all ones and zeros with size $K\times K$ respectively,
$\odot$ denotes Hadamard product and $[\mathbf{A}(k,n)]$ is $(k,n)^{th}$
element of matrix $\mathbf{A}$. Let $tr(\mathbf{X})$, $vec(\mathbf{X})$ denote the trace operation,and the column vector obtained
by stacking the columns of $\mathbf{X}$. $\mathcal{CN}(a,\mathbf{R})$ is used to denote circularly symmetric
complex Gaussian random vectors, which has the mean $a$ and
the covariance matrix $\mathbf{R}$. Finally, $\cup$ is the union of sets
and $f(x)$ denotes a function of $x$.
\vspace{-0.4cm}
\section{System model}
Our model consists of a network of $C$ time-synchronized
cells with full spectrum reuse, each one of the cells serves $L$ users. Estimation of flat block fading, narrow band
channels in the uplink is considered, and all the base stations
are equipped with an $M$-element uniform linear array (ULA) of antennas. We assume that the training sequences, of length $\tau$ symbols, used by single-antenna
users in the same cell are mutually orthogonal. However,  training sequences are reused in a multicell environment from cell
to cell. The training sequences used for estimating the user channels are denoted by $\mathbf{s}_i\triangleq\begin{array}{ccc}[s_{i1}&\hdots&s_{i\tau}]^T\end{array}\in\mathbb{C}^{\tau\times 1}$. The training sequence symbols are normalized such that $\{|s_{ij}|^2=\frac{P}{\tau},
\forall
j\in \tau\}$, where $P$ is the total training sequence power.
Channel vectors are assumed to be $\mathbb{C}^{M\times1}$ Rayleigh fading
with correlation due to the finite multipath angle spread seen from the base station side, the $l^{th}$ user's channel towards $c^{th}$ BS is given by $\mathbf{h}_{lc}\sim\mathcal{CN}(0,\alpha_{lc}\mathbf{R}_{lc})\in
\mathbb{C}^{M\times 1}$, where $\alpha_{lc}$ is the attenuation from the $l^{th}$ user to $c^{th}$
BS. We denote
the channel covariance matrix $\mathbf{R}_{lc}\in\mathbb{C}^{M \times M}$ as
$\mathbf{R}_{lc}=\mathbb{E}[\mathbf{h}_{lc}\mathbf{h}_{lc}^H]$.

Considering the transmission of $\mathbf{s}_i$
sequence, the $ M\times {\tau}$  signal baseband symbols sampled at the
$c^{th}$ target base station is
\vspace{-0.2cm}
\begin{eqnarray}
\mathbf{Y}_c=\displaystyle\sum_{i}\sum_{\forall l\in \mathcal{K}_i}\mathbf{h}_{lc}\mathbf{s}^T_i+\mathbf{N}_c.
\end{eqnarray}
where $\mathcal{K}_i$ is the set of users who use the training
sequence $\mathbf{s}_i$. $\mathbf{N}_c\in\mathbb{C}^{M\times \tau}$  is the spatially and temporally white
complex additive Gaussian noise (AWGN) with element-wise variance $\sigma^2$. We define a training matrix $\mathbf{S}_i=\mathbf{s}_i\otimes \mathbf{I}_M$, such that $\mathbf{S}^H_i\mathbf{S}_i=\tau\mathbf{I}_M$. Then, the received training
signal at the target base station can be expressed as
\begin{equation}
\label{rx}
\vspace{-0.2cm}
\mathbf{y}_c=vec(\mathbf{Y}_c)=\sum_{i}\mathbf{S}_i\displaystyle\sum^{C}_{\forall l\in
\mathcal{K}_i}\mathbf{h}_{lc}+\mathbf{n}_c
\end{equation}
where $\mathbf{n}_c\in \mathbb{C}^{M\tau\times 1}=\text{vec}(\mathbf{N})$ is the sampled noise at the $c^{th}$ BS. Since the pilots are orthogonal $\mathbf{s}^H_i\mathbf{s}_j=0$,
this makes $\mathbf{S}^H_i\mathbf{S}_j=\mathbf{0}^{M\times M}$. Due to orthogonality between different sequences, the sampled signal resulted
from the transmissions of the $i^{th}$ training sequence at $c^{th}$
BS can be isolated from other training sequences and can be expressed as
\begin{eqnarray}
\mathbf{y}_c=\mathbf{S}_i\displaystyle\sum_{\forall l\in \mathcal{K}_i}\mathbf{h}_{lc}+\mathbf{n}_c
\end{eqnarray}
 For the sake of simplicity and without lack
of generalization, we drop the training sequence
index and assume that a single training sequence is used over the network, which makes
the baseband signal sampled at $c^{th}$ base station as
\begin{equation}
\label{yc}
\mathbf{y}_c=\mathbf{S}\displaystyle\sum_{l=1}^{C} \mathbf{h}_{lc}+\mathbf{n}_c
\end{equation}
where
 $l=1,\hdots, C$ denotes the users transmitting the training sequence $\mathbf{s}$. Furthermore, we assume that there is time
synchronization in the system for coherent uplink transmissions.
\vspace{-0.6cm}
\subsection{Channel model}
We consider a uniform linear array (ULA) whose
response vector
can be expressed as
\begin{eqnarray}
\mathbf{a}(\omega)=\begin{array}{cccc}[1&e^{-j\omega}&\hdots&e^{-j(M-1)\omega}]^T\end{array}
\end{eqnarray}
where $\omega=\frac{2\pi d\sin\theta}{\lambda}$, $d$ is the antenna spacing at the base station,
$\lambda$ is the signal wavelength and $\theta$ is angle of arrival of single
path.
The received signal at the base station can be
expressed as a multipath model utilizing the response array vector as
\vspace{-0.2cm}
 \begin{eqnarray}
 \label{multipath}
 \mathbf{h}_{lc}=\sum_{i=1}^{Q}\gamma_i\mathbf{a}(\omega_i)
 \end{eqnarray}
 where $\gamma_{i}$ is complex random gain factor, $\theta_i$ is the angle
 of the arrival of the $i^{th}$ path, $Q$ is the number of paths.
 We
adopt a generic Toeplitz correlation model as it is the suitable model for
implementing the correlation from theoretical \cite{bjorn} and practical perspectives\cite{perz}. The two generic correlation types have a generic
Toeplitz structure.

In order to analytically study the performance of the proposed technique,
we deal with a simplified exponential correlation model due to its mathematical tractability
\cite{bjorn}. The correlation structure of $\mathbf{R}_{lc}$ can be formulated as following 
\subsubsection{Exponential Correlation\cite{bjorn}}
 It is known that exponential correlation matrix is a special case of Toeplitz and it is often
used for ULA
system, and it has the following formulation
\vspace{-0.1cm}
\begin{eqnarray}
\label{exp}
\small\hspace{-0.2cm}[\mathbf{R}(i,j)]=\begin{cases}\rho^{|i-j|},i>j\\
\rho^{|i-j|*},i<j
\end{cases}
 \end{eqnarray}\smallskip
 where $\rho\in \mathbb{C}, |\rho|\leq 1$. This kind of correlation is suitable for theoretical analysis, and it will be
 used in the next sections to study the benefits of the proposed framework.

 \subsubsection{Practical Correlation \cite{perz}}
 \label{practical}In order to approximate the practical correlation, the received signal at the base station can be
 implemented as limited memoryless multipath model with single tap utilizing the response array vector as (\ref{multipath}). \smallskip
  For a multipath scattering confined to a relatively small angular spread seen from the base station.
 A general correlation structure can be well approximated by
 \begin{eqnarray}
 \nonumber
 \mathbf{R}\approx\mathbf{D}_a\mathbf{B}\mathbf{D}^H_{a}
 \end{eqnarray}
where $\sigma_{\omega}=2\pi\frac{d}{\lambda}\sigma_{\theta}\cos\theta$, $\mathbf{D}_a=\text{diag}[\mathbf{a}(\omega)]$.\, $\sigma_{\theta}$ is
the standard deviation of the angular spread.\smallskip

$\mathbf{B}$ depends on the angular spread of the
multipath components. The angular distribution is Gaussian $\tilde{\omega}\in\mathcal{N}(0,\sigma_{\omega})$,
it can be written as
 \begin{eqnarray}
 \label{gaussian}
 [\mathbf{B}(m,n)]\simeq \exp({\frac{((m-n)\sqrt{3}\delta_{\omega})^2}{2}}).
 \end{eqnarray}

 For $\tilde{\omega}$ uniformly distributed over $[-\delta_{\omega},\delta_{\omega}]$,
 it has the following structure\smallskip
 \begin{eqnarray}
 \label{uniform}
 [\mathbf{B}(m,n)]\simeq\frac{\sin((m-n)\delta_{\omega})}{(m-n)\delta_{\omega}}.
 \end{eqnarray}
and $\sigma_{\omega}=\sqrt{3}\delta_{\omega}$. We adopt
the both correlation models: exponential and practical to test the efficiency
of the proposed algorithms. From the previous correlation expressions, it
can be argued that each user in the cell has a different covariance matrix due to
its position as the covariance matrix is a function of the angular spread
and its corresponding distribution.

The covariance information of
the target users and interfering users can be acquired   
exploiting resource blocks where the desired user and
interference users are known to be assigned training sequences
at different times. Alternatively, this information can be obtained using the
knowledge of the approximate users' positions and the type of the angular
spread at BS side exploiting the correlation equations (\ref{gaussian})-(\ref{uniform}).

\section{dct for spatial compression}
The optimal transform that decorrelates the signal is Karhunen-Lo\'{e}ve Transform (KLT) which requires
 significant computational resources \cite{klt}. Therefore, fixed transforms are preferred in
 many applications. In contrast to data-dependent transform KLT, DCT is a fixed
 transform that does not depend on the data structure and has excellent energy
 compaction properties that perform very close to KLT. Although there are
 several fixed transform techniques (i.e. FFT, DFT), DCT is more useful in
 the context of this paper since it has higher compression efficiency in
 comparison with the other techniques. DCT is a technique for converting a signal into elementary
frequency components.  It transforms the signal from time domain to frequency domain. Most of the signal information tends to be
concentrated in a few components of the DCT if the information is correlated \cite{dcto}. Therefore,
the signal can be compressed by keeping the important frequencies and
truncating the least influential ones without impacting the quality of estimate. DCT's qualities
motivate its utilization in time-domain estimation in OFDM-MIMO to reduce
the border effect owing to its capacity to reduce the high frequency components in the transform domain\cite{dct}.  In this work, we do not tackle
the multicarrier OFDM, we utilize the DCT to handle the interference at the
estimation process in multiple antenna multicell systems.\smallskip

The compression capability of DCT\footnote{See eq. (\ref{dctm})- (\ref{dctn}) for more details about DCT. 
} is depicted in Fig. 1 for a covariance
matrix as in (\ref{exp}) with $\rho=0.9$. The signal energy
is condensed in the first few spatial frequencies, and it can be noticed  that
$93\%$ of signal is compressed in the first four coefficients. The DCT capability compression
 helps in categorizing the important
and the unimportant spatial frequencies in order to concentrate the contamination
in fewer dimensions.
In order to evaluate the effectiveness of DCT compaction property, we study the correlation models
mentioned in the system model section.\smallskip

 In order to study numerically the compression capability of DCT, we note
 that $\mathbf{R}$ spans a set of eigenspaces which determines the directions
 of the transmissions.
 Each eigenspace acquires certain power, which is known as the eigenvalue  of this space. In this direction, we use the metric
 of the standard condition number (SCN), which is defined as the ratio of the maximum eigenvalue to the
 minimum eigenvalue and denoted by $\chi(\mathbf{R})$, before and after DCT
has been applied. The intuition behind using the SCN is to measure the dynamic
range of eigenvalues which is representative of the compression level.
This can be demonstrated the reduction of SCN, since preserving a small number of
strong eigenvalues means that the signal space can be effectively represented
by fewer eigenvectors. Thus, the same
energy amount is contained in lower dimensional subspace which indicates a compression. The maximum and the minimum eigenvalues for exponential correlation matrix
$\mathbf{R}$ as in eq. (\ref{exp}) are $\frac{1}{(1-\rho)^2}$ and $\frac{1}{(1+\rho)^2}$ respectively \cite{transform}.\smallskip
\begin{newtheorem}{lemma}{Lemma}\nonumber
\begin{lemma}
\label{lemmad}
\cite{transform}: The SCN of any exponential-form correlation matrix $\mathbf{R}$ is thus given by\smallskip
\begin{eqnarray}
\small\lim_{M\rightarrow \infty}\chi\text{ }(\mathbf{R})=\Big(\frac{1+\rho}{1-\rho}\Big)^2
\end{eqnarray}
The SCN can be very large for highly correlated data ($\rho\rightarrow
1$).\\
\end{lemma}
\vspace{-0.1cm}
\end{newtheorem}\smallskip
\begin{figure}[h]
\begin{center}
\label{figo}
\includegraphics[scale=0.65]{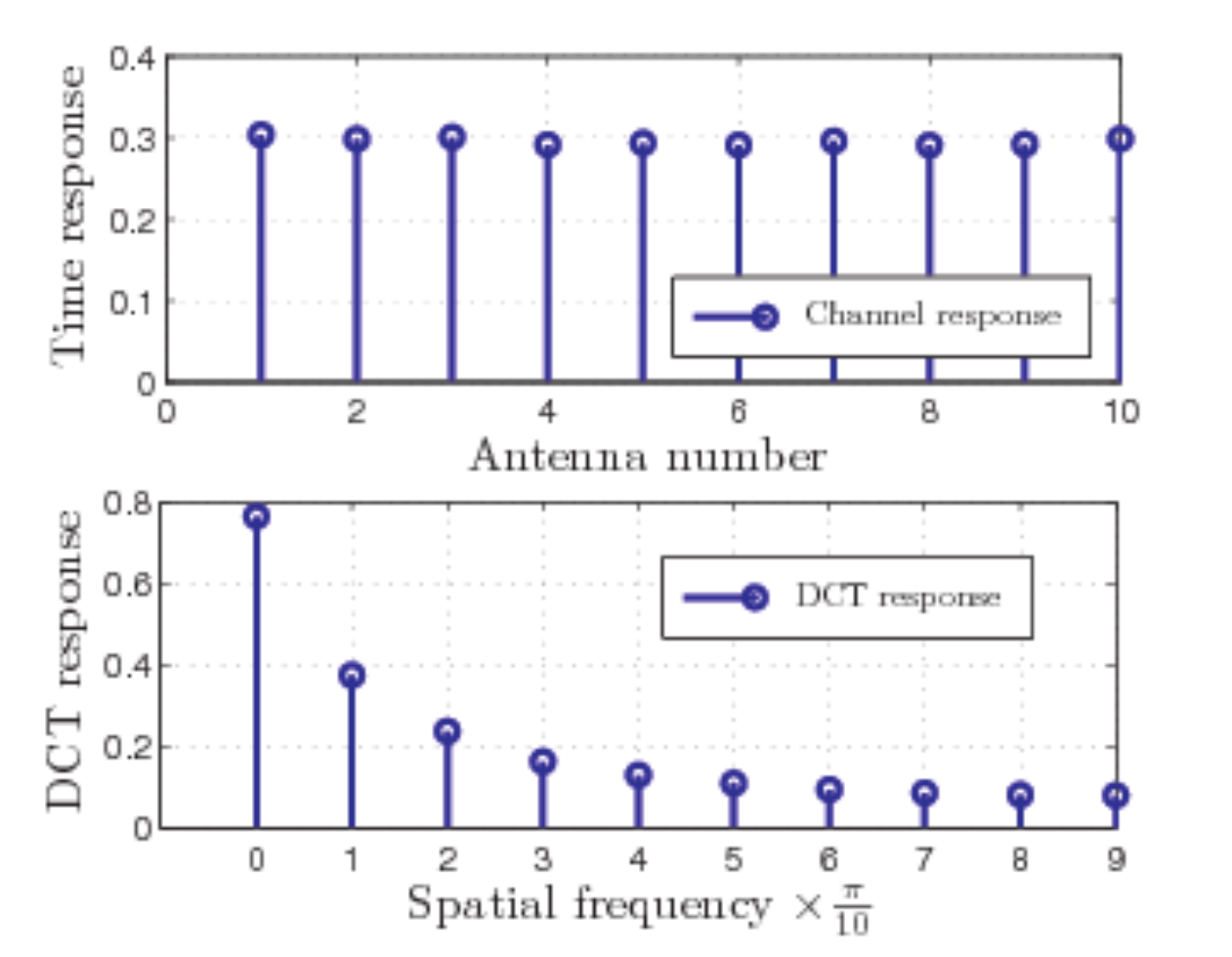}
\caption{\textit{\small Comparison between the DCT spatial frequency response
and spatial time response at the scenario of $\rho=0.9$ in (\ref{exp}).}}
\end{center}
\end{figure}
 The DCT basis comprises of the eigenvectors of the following symmetric tri-diagonal matrix\cite{cg}
 \begin{eqnarray}
\small \mathbf{Q}_c=\small\begin{bmatrix}\begin{array}{ccccc} 1-\zeta&-\zeta&0&\hdots&0\\
 -\zeta&1&0&\hdots&0\\
 \vdots&\ddots&\ddots&\ddots&\vdots\\
  \vdots&\ddots&-\zeta&1&-\zeta\\
 0&\hdots&0&-\zeta&1-\zeta\end{array}\end{bmatrix}
 \end{eqnarray}
 where $\zeta=\frac{\rho}{1+\rho^2}$. In matrix form, the eigenvectors are given by
\vspace{-0.2cm}
\begin{eqnarray}
\label{dctm}
[\mathbf{U}_D(k,n)]=c[k]\cos\bigg(\frac{(2n+1)k\pi}{2M}\bigg).
\end{eqnarray}
As a result, the definition of DCT has the following interpretation
\begin{eqnarray}
\label{dct}
\mathbf{m}_d[k]&=&c[k]\sum_{n=0}^{M-1}\mathbf{m}[n]\cos\Bigg(\frac{\pi(2n+1)k}{2M}\Bigg),
\end{eqnarray}
where $\mathbf{m}$ is the time domain vector and $\mathbf{m}_{d}$
is the DCT representation of $\mathbf{m}$
and can be written as:
\begin{equation}\label{dctn}
\mathbf{m}[n]=\sum_{k=0}^{M-1}c[k]\mathbf{m}_d[k]\cos\Bigg(\frac{\pi(2n+1)k}{2M}\Bigg),
\end{equation}
where $k$ represents the spatial frequency coefficient, $n$ denotes the antenna
number, and finally $c[k]$ is defined as below:
\begin{equation}
\nonumber
c[k]=\begin{cases}\sqrt{ \frac{1}{M}} \hspace{0.5cm}, k=0,\\
\sqrt{\frac{2}{M}}\hspace{0.5cm}, \text{otherwise}.
\end{cases}
\end{equation}

\vspace{-0.4cm}    
\subsection{DCT for Uniform Linear Array}
In order to interpret the DCT channel response, we need to find the DCT for ULA response, which can be expressed as (\ref{array}). From (\ref{array}), the DCT for channel vector following the expression in (\ref{multipath})
can be implemented as (\ref{arraydct}). To verify
the compression efficiency of employing DCT, we study its impact on the change
of eigenvalues of the exponential correlation matrix.

\begin{figure*}[t]
\hspace{0.2cm}
\begin{tabular}[t]{c}
\begin{minipage}{17 cm}
\nonumber
\small\begin{eqnarray}
\label{array}
\hspace{-0.6cm}
\mathbf{a}^{\text{DCT}}[k]=\small\begin{cases}\begin{array}{ccc}\small\sqrt{\frac{1}{M}}e^{-j\omega_i\frac{M-1}{2}}\frac{\sin \frac{M}{2}\omega_i}{\sin\frac{\omega_i}{2}}& &, k=0,\\
\sqrt{\frac{2}{M}}\frac{\sin(\small(\omega_i+\frac{\pi k}{M}\small)\frac{M}{2})}{\sin(\small(\omega_i+\frac{\pi
k}{M}\small
)\frac{1}{2})}e^{j\frac{(\omega_i+\frac{\pi k}{M})(M-1)}{2}+\frac{\pi k}{M}}+\sqrt{\frac{2}{M}}\frac{\sin(\small(\omega_i-\frac{\pi k}{M}\small)\frac{M}{2})}{\sin(\small(\omega_i-\frac{\pi
k}{M}\small
)\frac{1}{2})}e^{j\frac{(\omega_i-\frac{\pi k}{M})(M-1)}{2}-\frac{\pi k}{M}}&
&, 0\leq k\leq M-1.
\end{array}\end{cases}
\end{eqnarray}\smallskip

\begin{eqnarray}
\label{arraydct}
\hspace{-1.6cm}\mathbf{h}^{\text{DCT}}_{lc}[k]=\small\begin{cases}\begin{array}{ccc}\small\sum^P_{i=1}\frac{\gamma_i\sqrt{\frac{1}{M}}}{\sqrt{P}}e^{-j\omega_i\frac{M-1}{2}}\frac{\sin \frac{M}{2}\omega_i}{\sin\frac{\omega_i}{2}}& &,k=0,\\
\sum^P_{i=1}\frac{\gamma_i\sqrt{\frac{2}{M}}}{\sqrt{P}}\Bigg(\frac{\sin(\small(\omega_i+\frac{\pi k}{M}\small)\frac{M}{2})}{\sin(\small(\omega_i+\frac{\pi
k}{M}\small
)\frac{1}{2})}e^{j\frac{(\omega_i+\frac{\pi k}{M})(M-1)}{2}+\frac{\pi k}{M}}+\frac{\sin(\small(\omega_i-\frac{\pi k}{M}\small)\frac{M}{2})}{\sin(\small(\omega_i-\frac{\pi
k}{M}\small
)\frac{1}{2})}e^{j\frac{(\omega_i-\frac{\pi k}{M})(M-1)}{2}-\frac{\pi k}{M}}\Bigg)
& &,0\leq k \leq M-1.
\end{array}\end{cases}
\end{eqnarray}
\end{minipage}
\vspace{0.2cm}\\
\hline
\hline
\end{tabular}
\vspace{-0.5cm}
\end{figure*}

 \begin{lemma}\cite{transform}
 The SCN of the exponential correlation matrix eq.(\ref{exp}) after the DCT is such that
 \begin{eqnarray}
 \small\lim_{M\rightarrow \infty}\chi\text{}(\mathbf{U}_D\mathbf{R}{\mathbf{U}_D}^H)=1+\rho.
 \end{eqnarray}
 \end{lemma}\smallskip

The change of SCN from $\frac{1+\rho}{1-\rho}$ to $1+\rho$ before and after DCT transformation explains
the compressive nature of DCT.

The correlation matrix can be re-implemented as
$\mathbf{T}_{lc}=\mathbf{E}[\mathbf{m}^d_{lc}{\mathbf{m}^d_{lc}}^H]=\mathbf{A}_D\mathbf{R}_{lc}\mathbf{A}_D^{H}$, which is a 2-dimensional DCT of the $\mathbf{R}_{lc}$. Decontaminating the training sequences from the interference requires taking into the account to the following important
points\smallskip
\begin{itemize}
\item As most information is condensed in the certain frequencies, these frequencies should face the lowest possible interference in order to obtain an accurate estimate. The spatial frequencies are functions of angles of
arrival of pilots at the BS, thus condensing these spatial frequencies in
certain bands makes the separation process easier. This kind of separation can be enhanced by prearranging
the spatial frequencies through performing training
sequence allocation which is discussed later in section \ref{allocation}. \smallskip
\item The users who have the least common spatial characteristics should be assigned the same sequence; as they have the minimum overlap in the spatial frequency domain.\smallskip
\end{itemize}

The compression pattern of DCT is defined as a function of the power distribution
among the spatial frequencies. Such distributions depend on the nature of
the angular spread at the BS. To enable the spatial separation among the different estimate, the pattern of compression should be highlighted. For
this purpose, we
adopt the correlation models of \cite{per} to formulate generic rules
about DCT compression for ULA systems. \smallskip 
 
\begin{lemma}
\label{L1}
The spatial DCT frequencies are concentrated at low frequencies if angle of arrivals
are close to zero. If the direction of the arrivals are close to $\frac{\pi}{2}$,
the DCT components are condensed at high frequencies.
\end{lemma}
\begin{proof}
See Appendix.\smallskip
\end{proof}

 The previous Lemma \ref{L1} confirms the intuition that the contamination
is minimum when the important frequencies of the users who allocate the same
training sequence are concentrated at different spatial band. The compression efficiency enhances with increasing the number of antennas at the BS, which
makes the separation much easier.     

\begin{lemma}
As $M\rightarrow\infty$, the DCT response is condensed in a single spatial
frequency component.
\end{lemma}
\begin{proof}
See Appendix.
\end{proof}
In the next section, we utilize the DCT compression to deal with the problem of contaminated estimation in correlated multiuser environment.
\vspace{-0.25cm}
\section{Channel estimation with $\mathcal{K}$- training sequence reuse}
Exploiting the ULA structure, we develop a new estimator with the aim of decontaminating the reused training sequences over the network. Our estimators utilize the embedded information in the second order statistics of the channel vectors. The  covariance matrices capture the embedded information related to the distribution (mainly mean and spread) of the multi-path angles of arrival at the base station \cite{kammeryer}. In \cite{emil-o}-\cite{gershman}, the authors focus on optimal training sequence
designs
and they exploit the covariance matrices of the desired channels and interference. The optimal training sequences are developed with adaptation to the statistics of the disturbance\cite{emil-o}. An extension of \cite{emil-o} is proposed in \cite{emil-t}, which proposes a more general framework for the purpose
of training sequence design in MIMO systems, which handles not only minimization of channel estimator’s MSE as an optimization metric, but
also the optimization of a final performance metric of interest
related to the use of the channel estimate in the communication
system. However, the design of the training sequence does not  have an impact on  interference mitigation, as long as, we utilize fully aligned training
sequences. Here, we concentrate on designing an estimation technique that can achieve accurate results, by exploiting the spatial frequency to mitigate the interference in CSI estimation process.

\vspace{-0.2cm}
\subsection{ Bayesian Estimation}
\vspace{-0.1cm}
The Bayesian Estimator (BE) is widely discussed in the literature and is utilized in many applications \cite{mackay}. The BE coincides with the minimum mean square estimator in case $\mathbf{h}_{lc}$ and $\mathbf{y}_c$ are jointly Gaussian distributed random variables. The BE estimator can be formulated
as \cite{mackay}

 \begin{equation}
\label{mmse}
\mathbf{F}_{lc}=\mathbf{G}_{lc}\mathbf{S}^H=\mathbf{R}_{lc}\bigg(\displaystyle \bigg(\sum_{l=1}^{C}\mathbf{R}_{lc}\bigg)+\frac{{\sigma_n}^2}{\tau}\mathbf{I}_{M}\bigg)^{-1}\mathbf{{S}}^H.
\end{equation}
The previous formulation provides insights about the nature of the training
sequence contamination; the estimator is a function of all the correlation matrices related to all users who utilize the same sequence. Therefore, the spatial characteristics of all users, who have the same sequence, influence on the accuracy of the estimations.

\subsubsection{Minimum Mean Square Error Performance }
The considered performance is the mean square error (MSE) of the proposed estimator, and can be expressed as\cite{mackay}
\begin{equation}
\label{mse1}
\mathcal{E}_{lc}=\mathbb{E}_{\mathbf{h}_{lc}}\Bigg\{\|\mathbf{\hat{h}}_{lc}-\mathbf{h}_{lc}\|^2|\mathbf{\hat{h}}_{lc}\Bigg\}.
\end{equation}

The final formulation for MSE can be expressed as
\begin{equation}
\label{mse}
\mathcal{E}_{lc}=tr\bigg(\mathbf{R}_{lc}-\mathbf{R}_{lc}^2\Big(\displaystyle\sum_{m=1}^C\mathbf{R}_{mc}+\frac{\sigma_n^2}{\tau}\mathbf{I}_{M}\Big)^{-1}\bigg).
\end{equation}
The upper bound of the MSE of the BE can be written as\smallskip
\begin{eqnarray}
\mathcal{E}^{MI}_{lc}=tr\bigg(\mathbf{W}_{lc}\Big(C\mathbf{\Delta}_{lc}+\frac{\sigma^2_n}{\tau}\mathbf{I}_M\Big)^{-1}\mathbf{W}^H_{lc}\bigg)
\end{eqnarray}
where $\mathbf{R}_{lc}$ is decomposed as $\mathbf{W}_{lc}\mathbf{\Delta}_{lc}\mathbf{W}^H_{lc}$,
and the subscript $MI$ refers to the ``maximum interference scenario". This
scenario occurs when the set of users, who  utilizes the same training sequence, has identical spatial second order statistic and attenuation towards certain
BS. \smallskip

The lower bound of the MSE  of the BE is  as:
\begin{equation}
\mathcal{E}^{NI}_{lc}=tr\bigg(\mathbf{R}_{lc}\Big(\mathbf{R}_{lc}+\frac{\sigma_n^2}{\tau}\mathbf{I}_{M}\Big)^{-1}\bigg)
\end{equation}
where superscript $NI$ refers to the ``no interference scenario". This lower
bound can be achieved when the users span distinct subspaces and this condition should be satisfied $\{\mathbf{W}_{lc}\mathbf{W}_{lj}=\mathbf{0}
, \forall j\neq c \}$.\smallskip

 Therefore, the overlap in these subspaces will
degrade the estimate. A new look to the problem will be handled through the DCT framework in the next sections. As the work in this paper aims at minimizing the estimation errors to reduce the contamination, we do not consider beamforming which is handled in \cite{jose}. Taking this into the account, we only consider conventional beamforming techniques:
\begin{itemize}
\item Coordinated Beamforming (CB) requires the estimation of the intefering channels. The contamination occurs when the same BS assigns the same sequence to estimate the served user channel as well as interfering channels and/or other BSs use the same training to estimate their users' channel or the corresponding interfering channels. In this scenario, the contamination can
be utilized to estimate the interfering channel for the user that utilizes the same training
sequence.
\item Maximum ratio transmission (MRT), which requires the estimation of the desired user channel. The contamination happens when the base stations  utilize the same training sequence for their users.\smallskip
\end{itemize}
\vspace{-0.3cm}
\subsection{Least Square Estimation}
Least square estimation can be utilized in the scenarios when the information
about the second order statistics is not available.
Hence, if the
received signal is modeled as (\ref{rx}), a least square (LS) estimator for the desired channel can be formulated as\smallskip
\begin{eqnarray}
\vspace{-0.2cm}
\label{ls1}
\mathbf{\hat{h}}^{\text{ls}}_{lc}=\mathbf{S}^H\mathbf{y}_c.
\end{eqnarray}
The conventional estimator suffers from a lack of orthogonality
between the desired and interfering training sequences, an effect known as
training sequence contamination \cite{jose},\cite{larsson}-\cite{Jindal}. In particular, when the same
training sequence is reused in all cells,
the estimated channel can be expressed as
\begin{eqnarray}
\label{ls}
\vspace{-0.2cm}
\mathbf{\hat{h}}^{\text{ls}}_{lc}=\mathbf{h}_{lc}+\sum^C_{m\neq c}\mathbf{h}_{mc}+\frac{\mathbf{S}^H\mathbf{N}}{\tau}.
\end{eqnarray}
As it appears in (\ref{ls}), the interfering channels have strong impact and leak directly
into the desired channel estimate. The estimation performance
is then limited by the signal to interfering ratio at the base
station, which consequently limits the ability to design effective beamforming solutions.\smallskip
\subsubsection{Mean Square Error Performance}
 The MSE for LS can derived as (\ref{mse1}), and the closed form
expression can be formulated as
\begin{eqnarray}
\vspace{-0.5cm}
\label{msels}
\mathcal{E}^{LS}_{lc}=tr\bigg(\sum^C_{m=1,m\neq l}\mathbf{R}_{mc}\bigg).
\end{eqnarray}
From (\ref{msels}), it can be noted that MSE is a linear function of the involved correlation
matrices traces, and it grows linearly with the number of users who utilizes the
same training sequence. 
\section{Estimation techniques using DCT compression}

\subsection{Important Frequencies Determination}
In order to decontaminate the channel estimate, we need to determine the important
frequencies that should be extracted from the estimated channel $\mathbf{S}^H\mathbf{y}_c$. To determine these important frequencies, we need to solve the following quadratic
form
\vspace{-0.3cm}
\begin{eqnarray}
\omega^*_i=\underset{\omega_i^*}{\arg\max}\quad\mathbf{a}_{\text{D}}(\omega_i)\mathbf{R}_{lc}\mathbf{a}^H_{\text{D}}(\omega_i).
\end{eqnarray}
The solution lies in finding the eigenvectors related to the maximum eigenvalues.
 Therefore, the concentration of the important spatial frequencies depend on the angular spread at the BS. 
\smallskip
In order to separate the contamination from the required estimate, we need
to determine the number of the important frequencies. 
It should be noted that the number of considered spatial frequency depends
on the tradeoff  between the loss resulted from contamination
and DCT lossy compression. We define
the vector $\mathbf{q}_{lc}\in\{0,1\}^{M\times 1}$, which identifies the unimportant and important frequencies of $\mathbf{h}_{lc}$ as 0 and 1 respectively. A definition of the compression ratio of can be stated as
\vspace{-0.3cm}
\begin{eqnarray}
\vspace{-0.5cm}\eta=\frac{\|\mathbf{q}_{lc}\|_0}{M}.
\end{eqnarray}

 In the next section, integrated DCT-LS and DCT-BE frameworks are proposed to deal with the problem of contaminated estimation in correlated multiuser environment.
\vspace{-0.4cm}
\subsection{DCT Based Bayesian Estimation (DBE)}
As the received signal is a combination of all users' channel
who utilize the training sequence, the DCT of the received signal contains
the spatial frequencies related to these channels compressed in certain components.
Therefore, a splitting technique in DCT domain is applicable if we know the most important frequencies related to the required estimate. To reduce the
impact of contamination, we can extract these frequencies and replace the
least important frequencies by their average values, which are close to zero. This can be explained by Fig.(\ref{mdct})-a (the upper figure), it depicts the splitting capability of DCT by distinguishing the
important information of the involved estimate, which is clearly condensed at
high and low frequencies. For the first estimate (the black one), the information
is compressed in the low spatial frequency components while the high spatial
frequencies
do not hold much information. It can be noticed that these
spatial frequencies are contaminated by other channels (the blue and the
pink ones). This contamination can be tackled by removing these spatial frequencies
and equating them to zero.

\begin{figure}[h]
\vspace{-0.4cm}
\begin{center}
\includegraphics[scale=0.55]{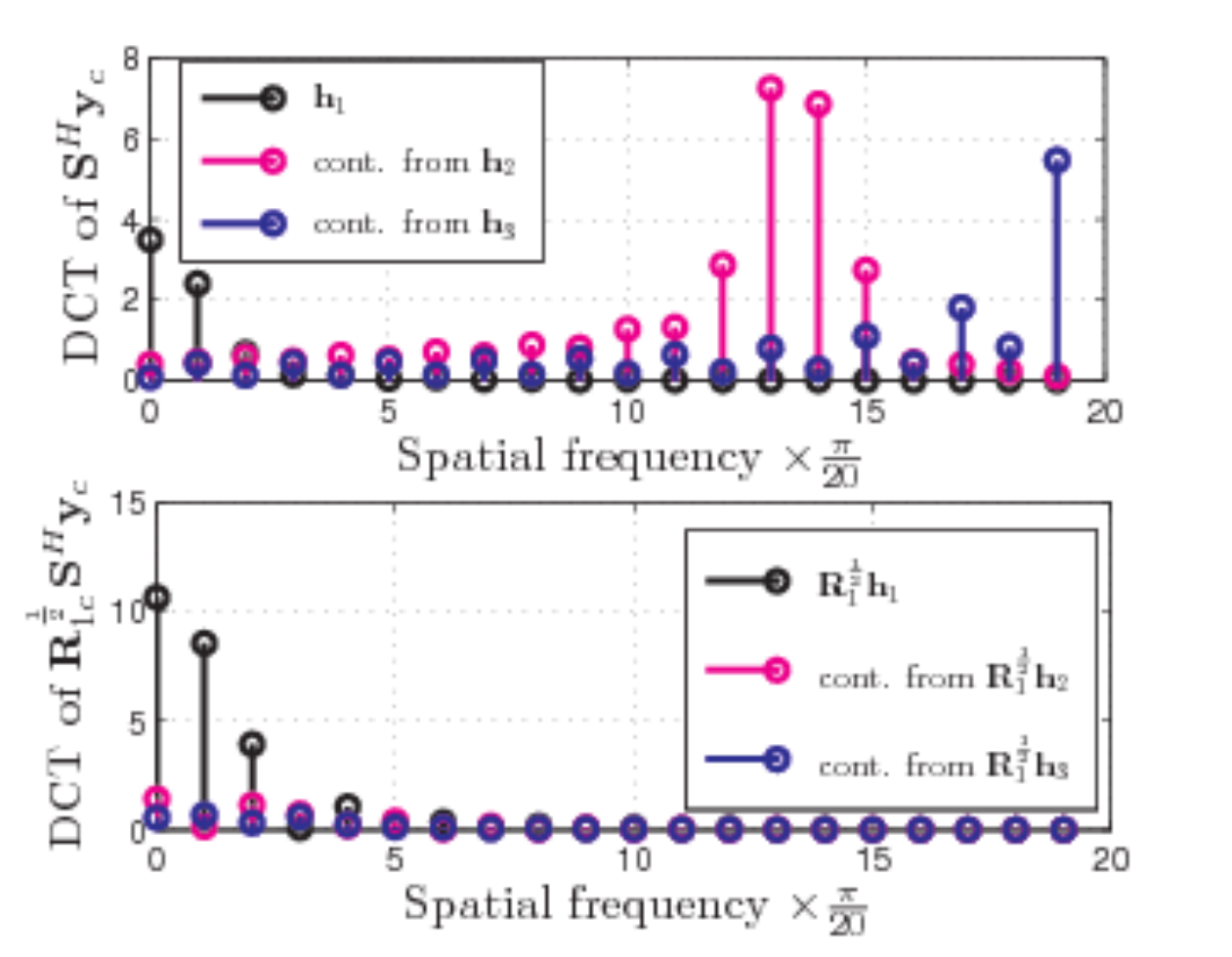}
\vspace{-0.4cm}\caption{\textit{\label{mdct}\small A comparison of DCT response
employing before and after multiplication by $\mathbf{R}^{\frac{1}{2}}_{lc}$ `cont' denotes contamination.}}
\end{center}
\end{figure}
The new estimation technique should take into consideration the covariance
information and DCT compression capability to boost the performance of BE.
The new estimation procedure can be described as following \smallskip

\footnotesize
\begin{center}
\begin{tabular}{p{8.3cm}}
\hline\hline
\textbf{A1.} DCT based Bayesian Estimation (DBE)\\
\hline
\begin{itemize}
\item At user's side: Send the training sequence $\mathbf{s}$.
\item At BS side:
\begin{enumerate}
\item Modify the covariance matrices in ${\mathbf{G}}_{lc}$ using $\hat{\mathbf{R}}_{lc}$,
which can be formulated as follows: 
\begin{eqnarray}\nonumber
\mathbf{\hat{U}}_D&=&\mathbf{U}_D\odot\big(\mathbf{q}_{lc}\otimes\mathbf{1}_{1\times
M}\big)\\\nonumber
\mathbf{U}_D\mathbf{\hat{R}}_{lc}\mathbf{U}^H_D&=&\mathbf{\hat{U}}_D\mathbf{R}_{lc}\mathbf{\hat{U}}^H_D
\end{eqnarray}
\item Find $\mathbf{\hat{y}}_c=\mathbf{S}^H\mathbf{y}_c$, extract the important
information related to the estimate $\tilde{\mathbf{y}}_c=(\mathbf{U}_D\hat{\mathbf{y}}_c)\odot\mathbf{q}_{lc}$.
\item Take the inverse DCT of $\tilde{\mathbf{y}}_c$ as $\mathbf{U}^H_D\tilde{\mathbf{y}}_c=\mathbf{y}^f_c$.
\item Find $\hat{\mathbf{h}}_{lc}=\mathbf{G}_{lc}\mathbf{y}^f_c$.  
\end{enumerate}
\end{itemize}\\
\hline
 \end{tabular}

 \end{center}
 \normalsize

\smallskip
\subsection{DCT Based Least Squares Estimation (DLS)}
\vspace{-0.1cm}
In comparison
with typical LS, in which the estimation results in direct summation of all
channel (\ref{ls}), the DCT Based LS (DLS) requires information about the important
frequency set for each estimated channel. The estimation
can be summarized as follows\\

\footnotesize
\begin{center}
\begin{tabular}{p{8.3cm}}
\hline\hline
\textbf{A2.} DCT based Least Square Estimation (DLS)\\
\hline
\begin{itemize}
\item At the user's terminals: Send the training sequence $\mathbf{s}$
\item At the BS's side:
\begin{enumerate}
\item Employ typical LS on the received signal, find $\hat{\mathbf{h}}^{ls}_{lc}$.
\item Employ DCT on $\hat{\mathbf{h}}^{ls}_{lc}$, extract the important DCT frequencies
related to estimate by $\tilde{\mathbf{h}}^{ls}_{lc}=(\mathbf{U}_D\hat{\mathbf{h}}^{ls}_{lc})\odot\mathbf{q}_{lc}$.\smallskip
\item Take the inverse DCT of the extracted version $\mathbf{U}^H_D\tilde{\mathbf{h}}^{ls}_{lc}$.\smallskip
\end{enumerate}
\end{itemize}\\
\hline
\end{tabular}
\end{center}
\normalsize

This estimation technique can be combined with a training sequence allocation algorithm
to make the most important spatial frequencies distinct which simplifies
the separation of these frequencies.
 The DLS can be used to estimate the channel for MRT
beamforming since it only requires the statistical information of the
target channel.\smallskip

The major difference between DBE and DLS is that DBE decontaminates the training
sequences
using two steps: zeroing the unimportant frequencies, and the BE step to remove the contamination from the important frequencies. BE and DBE require the covariance acquisition
of the direct and all interfering links. Another complexity source is the  matrix inversion. Thereby, their usage in large scale MIMO systems is limited
and constrained to the limited size problems. On the other hand,
DLS only requires the covariance knowledge of the target channel. It has less complexity due to the absence of a matrix inversion step and the acquisition
of the interfering links' covariance matrices. 
\vspace{-0.2cm}  
\section{Modified Spatial Estimation}
\label{modified}

To enhance the performance of BE, DBE and DLS, the important frequencies should
be more distinguishable. This can be obtained by compressing them into smallest
possible number of components to separate them more effectively. Therefore, a modification can be proposed
to handle this issue using the available correlation information and DCT.

The concept of increasing the correlation of the channel of interest
is proposed  in this section. This makes the target channel more ill-conditioned
and thereby the contamination more separable. Therefore, to estimate $\mathbf{h}_{lc}$
we can increase its correlation by multiplying the received signal $\mathbf{y}_c$
in (\ref{yc}) by $\mathbf{R}^{\frac{i}{2}}_{lc}\mathbf{S}^H$,
which makes the correlation of the target estimate $\mathbf{R}^{i+1}_{lc}$
and the correlation of contamination $\mathbf{R}^{\frac{i}{2}}_{lc}\mathbf{R}_{mc}\mathbf{R}^{\frac{i}{2}}_{lc}$.

The effect of implementing such step for $i=1$ is depicted in Fig.(\ref{mdct})-b. In the
considered scenario, $\theta_s=\{0^{\circ},15^{\circ},35^\circ\}$
are the lower limit of the angular spread for each channel respectively, $\Delta\theta=20^{\circ}$ is the angular spread. It can be noted that after employing this step, the impact of contamination in the DCT domain is
limited and does not have any influential contribution to the estimate and
can be removed easily using DCT. Taking into the account this property, we
propose two estimation techniques as follows     
 \begin{figure}[h]
\vspace{-0.1cm}
\begin{center}
\vspace{-0.1cm}
\includegraphics[scale=0.4]{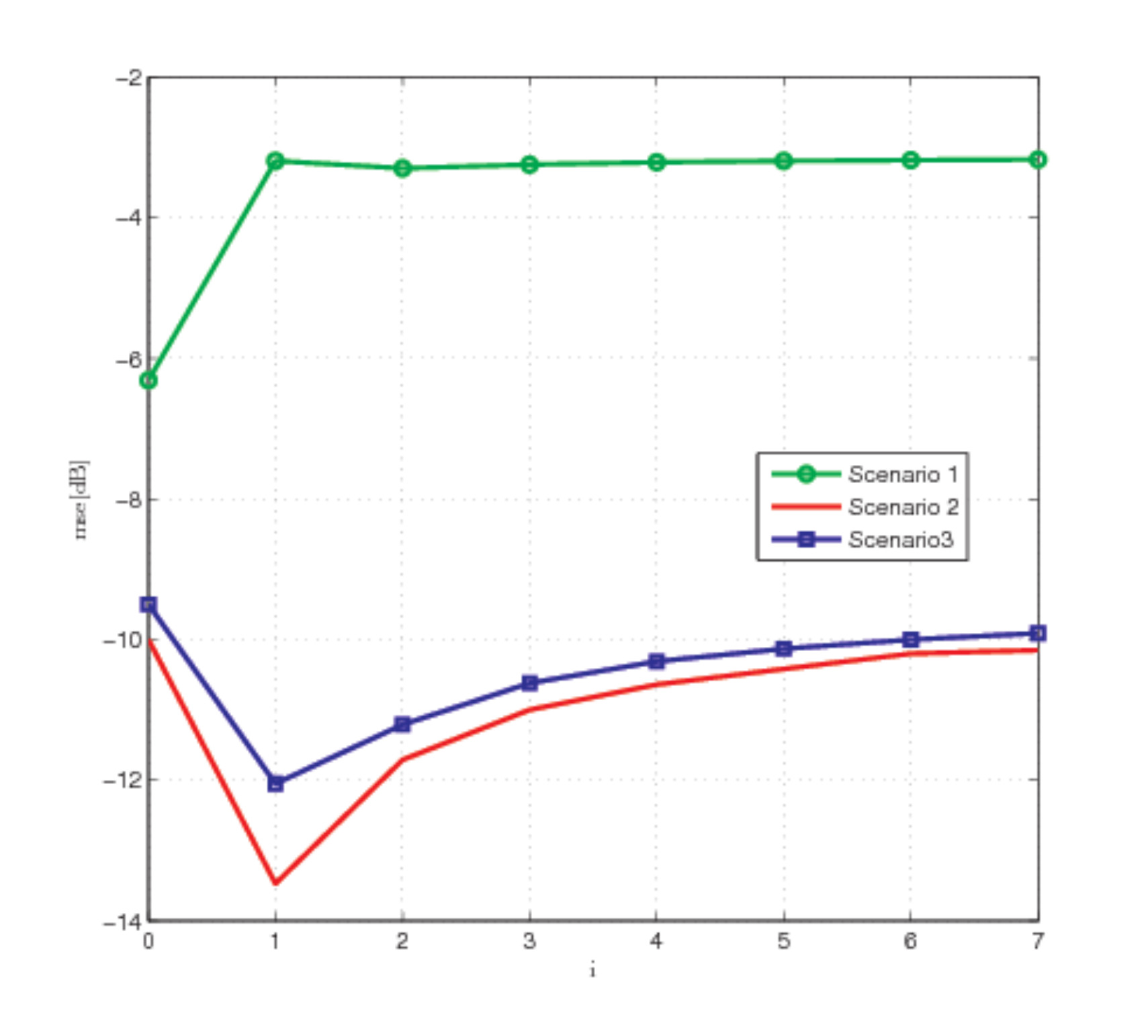}
\vspace{-0.4cm}\caption{\label{ant}\label{variants}\textit{\small The normalized MSE vs.
$\mathbf{R}^{\frac{i}{2}}_{lc}$ at different correlation scenarios, $K=3$,
$\Delta\theta=20^\circ$. Scenario
1:$\theta_s=[0^{\circ},20^{\circ}, 40^{\circ}]$ , Scenario 2:$\theta_s=[0^{\circ},15^{\circ}, 30^{\circ}]$, Scenario 3:$\theta=[0^{\circ}, 10^{\circ}, 20^{\c3irc}]$ }}
\end{center}
\end{figure}
\vspace{-0.1cm}
\subsection{Modified BE Based Estimation (MBE)}
\textcolor{black}{The   multiplication by $\mathbf{R}^{\frac{i}{2}}_{lc}$ increases the correlation
of the target estimate,
the modified correlation can be decomposed as $\mathbf{W}_{lc}\mathbf{\Delta}^{i
+1}_{lc}\mathbf{W}^H_{lc}$. This makes eigenvalues of the target estimate
more distinct since the stronger eigenvalue becomes more effective and the
opposite holds for the weaker ones. Therefore, the channel  becomes more ill-conditioned and
more separable. However, the impact of this multiplication is unknown with
respect to the interfering subspaces and heuristic approach is proposed to
gain the benefit of it when it is applicable. After the multiplication, the received signal can
be written as
\begin{eqnarray}
\mathbf{y}^{'}_c=\gamma_{lc}\mathbf{R}^{\frac{i}{2}}_{lc}\mathbf{S}^H\mathbf{y}_c=\gamma_{lc}\mathbf{R}^{\frac{i}{2}}_{lc}\mathbf{S}^H\mathbf{S}\sum^C_{m=1}\mathbf{h}_{lc}+\mathbf{R}^{\frac{i}{2}}_{lc}\mathbf{S}^H\mathbf{n}_c
\end{eqnarray}
  where $\gamma_{lc}$ is designed to keep the signal power fixed $\|\mathbf{y}^{'}_c\|^2=\|\mathbf{y}_c\|^2$. The modified BE can be derived in a similar fashion to traditional BE (19)-(23)
in the original manuscript but
with assuming that correlation matrices $\mathbf{Q}_{mc}=\mathbf{R}^{\frac{i}{2}}_{lc}\mathbf{R}_{mc}\mathbf{R}^{\frac{i}{2}}_{lc}$.
To utilize the enhanced compression capability, multiply the received signal by $\mathbf{R}^{\frac{i}{2}}_{lc}\mathbf{S}^H$
and then employ the following filter to estimate the time domain 
$\mathbf{h}_{lc}$\footnote{It is derived in a similar fashion to Bayesian
estimation technique using different correlation set $\mathbf{R}^m_{lc}=\mathbf{R}^2_{lc}$,
the estimated channel for the $c^{th}$ user and $\mathbf{R}^m_{lm}=\mathbf{R}^{\frac{1}{2}}_{lc}\mathbf{R}_{lm}\mathbf{R}^{\frac{1}{2}}_{lc}$.  }. This makes the formulation of the final estimation at the BS as}

 \footnotesize
\begin{eqnarray}
\label{MBE}
\hspace{-0.2cm}\mathbf{P}^i_{lc}=\gamma_{lc}\mathbf{R}^{1+\frac{i}{2}}_{lc}\Bigg(\mathbf{R}^{\frac{i}{2}}_{lc}\bigg(\sum^{C}_{m=1}\mathbf{R}_{mc}\bigg)\mathbf{R}^{\frac{i}{2}}_{lc}+\frac{\sigma^2_n}{\tau}\mathbf{R}^i_{lc}\Bigg)^{-1}\mathbf{R}^{\frac{i}{2}}_{lc}\mathbf{S}^H\end{eqnarray}
\normalsize
The total MSE of estimating $\mathbf{h}_{lc}$ can be evaluated as
\begin{eqnarray}
\epsilon^{i}_{lc}=\mathbb{E}_{\mathbf{h}_{lc}}\Big\{tr\Big((\mathbf{h}_{lc}-\mathbf{P}^i_{lc}\mathbf{y}_c)(\mathbf{h}_{lc}-\mathbf{P}^i_{lc}\mathbf{y}_c)^H\Big)\Big\}
\end{eqnarray} 
and finally the MSE can have the following closed form expression
\footnotesize
\begin{eqnarray}
\label{errormse}
\epsilon^{i}_{lc}=tr\Bigg(\mathbf{R}_{lc}-\gamma^2_{lc}\mathbf{R}^{{i+2}}_{lc}\Big(\mathbf{R}^{\frac{i}{2}}_{lc}\big(\sum^C_{m=1}\mathbf{R}_{mc})\mathbf{R}^{\frac{i}{2}}_{lc}+\frac{\sigma^2_n}{\tau}\mathbf{R}^{i}_{lc}\Big)^{-1}\Bigg).
\end{eqnarray}
\normalsize

To study the impact of this multiplication on the MSE of different MBE is studied in Fig.(\ref{variants}),
in which the MSE is depicted with respect to $i$, where $i=0$ refers to typical
BE. It can be concluded at scenario 2 and 3 that this multiplication reduces
the MSE, which motivates the utilization of this step at the estimation process.
On contrary to the two scenarios, the first one shows a contradicting
performance. This can be explained that without the multiplication, it is
impossible to separate the channel estimations due to the complete overlap
among the users' subspaces. This factor passively enhances in case of this multiplication
due to the noise amplification factors. Therefore, an adaption between the
typical BE and modified BE is required to get the joint benefits of the both
techniques. Moreover, it can be noted that
the best variant (the $i^{th}$ power) of $\mathbf{R}_{lc}$ from the MSE perspective
is the first one and it is sufficient to be used in the estimation.\smallskip
 
In order to optimize the system performance, we can adapt the estimation strategy
based on the involved users who utilize the same training sequence. The 
adaptive strategy
$\mathcal{J}_l$ can be expressed as 
\vspace{-0.3cm}
\begin{eqnarray}
\label{adaptive}
\mathcal{J}_l=
\begin{cases}\begin{array}{cc}
BE &,\mathcal{E}^{BE}_{lc}\leq\mathcal{E}^{MBE}_{lc}\\
MBE &,\mathcal{E}^{BE}_{lc}\geq\mathcal{E}^{MBE}_{lc}.
\end{array}\end{cases}
\end{eqnarray}
The selection procedure is implemented at the system installation level and
is performed once for each channel. Therefore, the users who utilize the same training
sequence can have different estimation techniques. It can be noted that the
important frequencies of the target estimate do not change with the multiplication
of $\mathbf{R}^{\frac{1}{2}}_{lc}$, since the eigenvectors of the estimate
are not change by this multiplication. However, the subspaces spanned by
interfering signals change due to the modified covariance matrices $\mathbf{\hat{R}}_{lm}=\mathbf{R}^{\frac{1}{2}}_{lc}\mathbf{R}_{lm}\mathbf{R}^{\frac{1}{2}}_{lc}$. However, there is
no mathematical explanation clarifies the relation between the eigenvectors of
$\mathbf{R}_{lm}$ and $\mathbf{\hat{R}}_{lm}$. This makes the selection of
the most appropriate estimation technique hard, thus the adaptation based on MSE is the solution for this issue.  
The adaptation depends on the subspaces spanned by the target and interfering signals related to traditional and modified techniques. Moreover, the multiplication
by other variants can lead to excessive contamination correlation with target
subspaces, this makes the contamination subspaces span the same subspaces as the target estimate which actually makes the contamination mitigation
much harder.

\vspace{-0.15cm}     
\subsection{Modified DCT Bayesian Estimation (MDBE)}
\normalsize
The concept of integrating the DCT with modified BE can be adopted to enhance
the estimation quality, the modified DCT BE can be described as following\\

\footnotesize
\begin{center}
\begin{tabular}{p{8.3cm}}
\hline\hline
\textbf{A3.} Modified DCT Bayesian estimation (MDBE) \\
\hline
\begin{itemize}
\item Multiply $\mathbf{R}^{\frac{1}{2}}_{lc}\mathbf{S}^H$ by the received
signal $\mathbf{y}_c$ and get $\mathbf{g}_{lc}$. 
\item Employ the same steps (1)-(4) as DBE at $\mathbf{g}_{lc}$ to $\hat{\mathbf{g}}_{lc}$.
It should be noted that $\mathbf{G}_{lc}$  should be replaced $\mathbf{P}_{lc}$.
\item To obtain the estimate $\mathbf{\hat{h}}_{lc}=\mathbf{\hat{g}}_{lc}$.
\end{itemize}\\
\hline
\end{tabular}
\end{center}
\normalsize
The adaptation idea can be exploited as in (\ref{adaptive}) to enhance the performance of DCT based
estimation to optimally select the estimation technique from DBE or MDBE.
\vspace{-0.1cm}
\subsection{Modified DCT LS Estimation (MDLS)}
\vspace{-0.1cm}
The joint utilization of the DCT and LS within the modified framework is proposed
to boost the quality of LS estimate, and can be summarized as follows\\
 \footnotesize
\begin{center}
\begin{tabular}{p{8.3cm}}
\hline\hline
\textbf{A4.} Modified DCT LS estimation (MDLS)\\ 
\hline
\begin{itemize}
\item Multiply $\mathbf{R}^{\frac{1}{2}}_{lc}\mathbf{S}^H$ by the received
signal $\mathbf{y}_c$ and get $\mathbf{g}_{lc}$. 
\item Employ the same steps (1)-(4) as DLS at $\mathbf{g}_{lc}$ to $\hat{\mathbf{g}}_{lc}$.
\item To obtain the estimate $\mathbf{\hat{h}}_{lc}=\mathbf{R}^{-\frac{1}{2}}_{lc}\mathbf{\hat{g}}_{lc}$.
\end{itemize}\\
\hline
\end{tabular}
\vspace{-0.1cm}
\end{center}\smallskip
\normalsize
rrelation matrix structure
$\mathbf{R}_{lc}=\mathbf{W}_{lc}\mathbf{\Delta}_{lc}\mathbf{W}^H_{lc}=\mathbf{\hat{W}}_{lc}\mathbf{\hat{\Delta}}_{lc}\mathbf{\hat{W}}_{lc}$,
where $\mathbf{\hat{\Delta}}_{lc}\in\mathbb{R}_{+}^{r
\times r}$,$\mathbf{\hat{W}}_{lc}\in\mathbb{C}^{M\times r}$ are the matrices that contain the positive eigenvalues and the eigenvectors that are associated
with the positive eigenvalues. This enables the implementation of the proposed
algorithms in the case of ill-conditioned channels. The usage of DLS and MDLS can be adapted like (\ref{adaptive})  to enhance the quality of estimation.  
\vspace{-0.2cm}
\section{Joint spatial estimation and training sequences allocation}
\label{allocation}
\vspace{-0.1cm}
Several training sequence assignment techniques are investigated in the recent
literature \cite{nossekt}-\cite{haifan}, in order to select the optimal
set of users that can utilize the same training sequence simultaneously.
These allocation techniques mimic the previously proposed scheduling algorithms in multiuser MIMO
scenarios\cite{yoo} as they try to assign the same training sequence to the users who have distinct
subspaces. 
 The training sequence assignment strategies inherit the concept of multiuser diversity,
which depends on the pool of the involved users, as a consequence, the system performance
is degraded for low number of users. In this section, an heuristic approach assigns
the available training sequences to required  users'
channel estimates is proposed as
\vspace{-0.35cm}
\subsection{Training Sequence Allocation for BE related techniques}
\normalsize

Define the set of users who utilizing the training sequence $\mathbf{s}_l$
by $\mathcal{U}(\mathbf{s}_l)$, and the set of user who adopt the modified
BE estimation as $\mathcal{U}_M(\mathbf{s}_l)\subset\mathcal{U}(\mathbf{s}_l)$
and the set of users whose channel are estimated by typical BE $\mathcal{U}_B(\mathbf{s}_l)\subset\mathcal{U}(\mathbf{s}_l)$
such that $\mathcal{U}_M(\mathbf{s}_l)\cup\mathcal{U}_B(\mathbf{s}_l)=\mathcal{U}(\mathbf{s}_l)$ and the set of all users $\mathcal{U}$. The
set of all users who have not assigned a training sequence yet is defined
as $\mathcal{M}$. We define a measuring function as (\ref{errormetric}), the allocation algorithm is described as 

\begin{figure*}[t]
\begin{tabular}[t]{c}
\begin{minipage}{17 cm}
\small
\begin{eqnarray}
\label{errormetric}
\hspace{-4cm}\small&\mathcal{E}(\mathcal{U}_B(\mathbf{s}_l),\mathcal{U}_M(\mathbf{s}_1))&=\sum_{\forall
c\in \mathcal{U}_M(\mathbf{s}_l) }tr\Bigg(\mathbf{R}_{lc}-\mathbf{R}^{\frac{3}{2}}_{lc}\bigg(\sum_{\forall
 m \in \mathcal{U}_M(\mathbf{s}_l)}\mathbf{R}^{\frac{1}{2}}_{lc}\mathbf{R}_{mc}\mathbf{R}^{\frac{1}{2}}_{lc}+\frac{\sigma^2_n}{\tau}\mathbf{R}_{lc}+\zeta\mathbf{I}_M\bigg)^{-1}\Bigg)\\\nonumber
&+&\sum_{\forall
c \in \mathcal{U}_B(\mathbf{s}_l) }tr\bigg(\mathbf{R}_{lc}-\mathbf{R}_{lc}^2\Big(\displaystyle\sum_{\forall
m\in\mathcal{U}_B(\mathbf{s}_l)}\mathbf{R}_{mc}+\frac{\sigma_n^2}{\tau}\mathbf{I}_{M}\Big)^{-1}\bigg).
\end{eqnarray}
\normalsize
\end{minipage}
\vspace{0.2cm}\\
\hline
\hline
\end{tabular}
\end{figure*}
\footnotesize
\begin{center}
\begin{tabular}{p{8.3cm}}\\
\hline
\hline
\textbf{A5.} Greedy Training Sequence Allocation Algorithm for (BE/MBE) estimation \\
\hline
\begin{itemize}
\item Start from any training sequence, without any loss of generality we
start with $\mathbf{s}_1$.
\item Determine the training sequence reuse factor $\mathcal{K}$.
\item Define the set of users who should allocate the same training sequence
as
$\mathcal{U}(\mathbf{s}_l)$.
\item Pick random user $u\in \mathcal{M}=\mathcal{U}-\cup^{l-1}_{n=1}\mathcal{U}(\mathbf{s}_n)$
\item Step 1: set $\mathcal{U}(\mathbf{s}_l)=u$
\item Step 2: if $|\mathcal{U}(\mathbf{s}_l)|\leq\mathcal{K}$ 
\begin{enumerate}
\item for users= $\mathcal{M}(1)$ to $\mathcal{M}(|\mathcal{M}|)$
\begin{eqnarray}\nonumber
u^*&=&\underset{u}{\text{min}}\Big(\underset{{u_M}\in\mathcal{M}}{\text{min}}\hspace{0.1cm}{\mathcal{E}({\mathcal{U}_B(\mathbf{s}_l),\mathcal{U}_M(\mathbf{s}_l)\cup \{u_M\})}}\\\nonumber
&,&\underset{{u_B}\in\mathcal{M}}{\text{min}}\hspace{0.1cm}{\mathcal{E}({\mathcal{U}_B(\mathbf{s}_l)\cup
\{u_B\},\mathcal{U}_M(\mathbf{s}_l))}}\Big)
\end{eqnarray}
\item$\mathcal{M}=\mathcal{M}/\{u^*\}$
\item Estimation technique selection for $u^*$  is done according to \footnotesize\begin{eqnarray}\nonumber\begin{cases}
\text{MDBE},
{\mathcal{E}({\mathcal{U}_B(\mathbf{s}_l),\mathcal{U}_M(\mathbf{s}_l)\cup \{u^*\})}}\leq{\mathcal{E}({\mathcal{U}_B(\mathbf{s}_l)\cup
\{u^*\},\mathcal{U}_M(\mathbf{s}_l)\big)}}\\
\text{DBE},{\mathcal{E}({\mathcal{U}_B(\mathbf{s}_l),\mathcal{U}_M(\mathbf{s}_l)\cup \{u^*\})}}\leq{\mathcal{E}({\mathcal{U}_B(\mathbf{s}_l)\cup
\{u^*\},\mathcal{U}_M(\mathbf{s}_l)\big)}}
\end{cases}
\end{eqnarray}
\normalsize
\item $\mathcal{U}(\mathbf{s}_l)=\mathcal{U}(\mathbf{s}_l)\cup \{u^*\}$
\item go to step 2.
\end{enumerate}
\end{itemize}\\
 \hline
 \end{tabular}
 \vspace{-0.1cm}
\end{center}
\normalsize

This algorithm aims at minimizing the estimation error in two steps: the
first step is summarized by allocating
training sequences to the set of users whose spatial signatures have the most distinctive
characteristics. The performance improves with the number of user as it
becomes more likely to find users with distinct second order
statistics to be assigned the same training sequence. The second step is
the selection of the optimal estimation technique whether it is the typical
or the modified BE. 
\vspace{-0.4cm}
\subsection{Training Sequence Allocation for DLS/MDLS Techniques}
By taking a look at (\ref{msels}), it can be noted the MSE depends on the
trace of the interfering covariance matrices of the users using the same
sequences. 
\begin{newtheorem}{theo1}{Theorem}
\begin{theo1}
For typical LS, employing any training sequence allocation  algorithm does
not reduce the MSE performance.\smallskip  
\vspace{-0.2cm}
\end{theo1}
\begin{proof}
It can be proved from (\ref{msels}) that MSE performance depends on the traces
of all involved covariance matrices and it does not depend on the signal
space spanned by each correlation matrices. Therefore any set of correlation
matrices has the same trace will result in the same MSE.\smallskip
\end{proof}
\end{newtheorem}

However, this theorem does not apply for DCT based LS estimation. Since the
MSE metric of LS does not provide us an indication about the spatial separation
among the subspaces and the spatial frequencies. Since the correlation matrices
capture the information about subspaces, the intersection of these subspaces
results in interference. Therefore, to measure the spatial separation
between the $l^{th}$ and $m^{th}$  users towards the $c^{th}$ BS, we define
the following metric:
\begin{eqnarray}
\label{sps}
\delta^c_{lm}=\frac{tr\big(\mathbf{R}_{lc}\mathbf{R}_{mc}\big)}{tr(\mathbf{R}_{lc})tr(\mathbf{R}_{mc})}
\vspace{-0.2cm}
\end{eqnarray}
where $0\leq\delta^c_{lm}\leq 1$. When $\delta^c_{lm}$ is close to 1, the
users span highly overlapped subspaces, but when $\delta^c_{lm}$ is close
to 0, the users span a highly separated subspaces. The spatial separation between the $m^{th}$
user in $c^{th}$ and interference subspaces of all interfering users at the
estimation process can be written as 
\vspace{-0.1cm}
\begin{eqnarray}
\label{spss}
\delta^c_{l}=\frac{tr\big(\mathbf{R}_{lc}\sum^C_{m=1,m\neq l}\mathbf{R}_{mc}\big)}{tr(\mathbf{R}_{lc})tr(\sum^C_{m=1,m\neq
l}
\mathbf{R}_{mc})}.
\end{eqnarray}

 Define the set of users who utilizing the training sequence $\mathbf{s}_l$
by $\mathcal{U}(\mathbf{s}_l)$, the set of all users $\mathcal{U}$ and the
set of all users who have not assigned a training sequence yet as $\mathcal{M}$. We define the following metric as 
\begin{eqnarray}\nonumber
\delta(\mathcal{U}(\mathbf{s}_l))=\displaystyle\sum_{l\in\mathcal{U}(\mathbf{s}_l)} \frac{tr\Big(\mathbf{R}_{lc}\sum_{m\in \mathcal{U}(\mathbf{s}_l),m\neq l}\mathbf{R}_{mc}\Big)}{tr\Big(\mathbf{R}_{lc}\Big)tr\Big(\sum_{m\in\mathcal{U}(\mathbf{s}_l),m\neq
c}\mathbf{R}_{mc}\Big)}.
\end{eqnarray}
The algorithm can be summarized in (\textbf{A6}).
\vspace{-0.2cm}
\begin{center}
\footnotesize
\begin{tabular}{p{8.3cm}}\\
\hline\hline
\textbf{A6.} Greedy Training Sequence Allocation Algorithm for DLS/MDLS\\
\hline
\begin{itemize}
\item Start from any training sequence. Without any loss of generality, we
start with $\mathbf{s}_1$.
\item Determine the training sequence reuse factor $\mathcal{K}$.
\item Define the set of users who should allocate the same training sequence
as
$\mathcal{U}(\mathbf{s}_l)$
\item Pick random user $u\in \mathcal{M}=\mathcal{U}-\cup^{l-1}_{n=1}\mathcal{U}(\mathbf{s}_n)$
\item Step 1: set $\mathcal{U}(\mathbf{s}_l)=u$
\item Step 2: if $|\mathcal{U}(\mathbf{s}_l)|\leq\mathcal{K}$ 
for users $\mathcal{M}(1)$ to $\mathcal{M}(|\mathcal{M}|)$
\begin{enumerate}
\item $u^*=\underset{{u}\in\mathcal{M}}{\text{minimize}}\hspace{0.1cm}{\delta({\mathcal{U}(\mathbf{s}_l)\cup \{u\})}}$
\item $\mathcal{M}=\mathcal{M}/\{u^*\}$
\item $\mathcal{U}(\mathbf{s}_l)=\mathcal{U}(\mathbf{s}_l)\cup \{u^*\}$
go to Step 1.
\end{enumerate}\vspace{-0.4cm}
\end{itemize}\\
\hline
\end{tabular}
\end{center}
\normalsize
 \smallskip

 However, the performance of DLS and MDLS  can be optimized using training
 sequences allocation
 algorithm. This can be explained by their dependency on the spatial content
 of the users who utilize the same training sequence. Unfortunately, as theorem
 1 states, the MSE metric of LS is incapable of addressing such a separation. On
 the other hand, the MSE of BE metric offers such a quality so we can optimize
 the performance of DLS and MDLS by allocating the training sequences as in previous greedy algorithms. 
\vspace{-0.1cm}    
\section{Numerical Results}
\label{sim}
In the previous sections, several algorithms were developed to decontaminate
the training sequences in multiantenna multicell interference channels. In this section, the
performance of these algorithms are evaluated and compared. Monte Carlo simulations are performed to examine the efficiency of the  proposed algorithms. To assess the system
performance, we need to compare it with the state of the art techniques as
 typical Bayesian estimator and least square estimator. The considered algorithms
 and the corresponding equations
 are summarized in table \ref{table1}. \smallskip

We consider a multi-cell network where the
users are all distributed around the base stations. We denote the ratio of
 the power of the direct link $\alpha_{cc}$ over interfering link $\alpha_{lc}$ with $\beta_{lc}=\frac{\alpha_{cc}}{\alpha_{lc}}$. To test the validity of
the proposed algorithms, we deal with estimation at interference limited
scenario, we consider $\beta_{lc}=1,\forall l,c$, we drop the index $\beta$ for the
ease of notation. We adopt the model of a
cluster of synchronized and hexagonally shaped cells. We assume that the training
sequence $P=0 dB$. The used metric for evaluating the system performance is normalized
sum mean square
error and can be expressed as: \\
\begin{equation}
\epsilon= 10 \log_{10}\bigg(\frac{\sum_{c=1}^C\|\hat{\mathbf{h}_{lc}}-\mathbf{h}_{lc}\|^2}{\sum_{c=1}^c\|\mathbf{h}_{lc}\|^2}\bigg).
\end{equation}
\begin{center}
\begin{table}
{\begin{tabular}{|p{2cm}|p{4.5cm}|p{0.5cm}|}
\hline
LS & Least Square Estimator&(\ref{ls1}).\\
\hline
\small BE& Bayesian Estimator&(\ref{mmse}).\\
\hline
\small DBE& DCT Bayesian Estimator&(\textbf{A1}).\\
\hline
\small MBE& Modified Bayesian Estimator&(\ref{MBE}).\\
\hline
\small DLS& DCT Least Square Estimator&(\textbf{A2})\\
\hline
\small MDBE& Modified DCT Bayesian Estimator&(\textbf{A3}).\\
\hline
\small MDLS& Modified DCT Least Square Estimator&(\textbf{A4}).\\
\hline
\small ABE-MBE& Adaptive BE-MBE&(\ref{adaptive}).\\
\hline
\small ADBE-MDBE& Adaptive DBE-MBE&\dots\\
\hline
\small PA& Training sequence Allocation& \small\textbf{A5}-\textbf{A6}.\\
\hline
\small WPA& Without training sequence allocation&\dots\\
\hline
\end{tabular}}\\

\caption{\label{table1}The acronyms used in the simulation section for the adopted algorithm
and their corresponding equations.}
\end{table}
\end{center}
\vspace{-0.2cm}

In this section, we adopt the correlation model in \ref{practical}  $\theta_{sk}$ denotes the lower limit of the angular spread for each user with respect to the $k^{th}$ BS. $\Delta\theta$ is the span of the angular spread and it is assumed to be
the same for all users.  $\Delta_o\theta$ determines
the amount of the overlap between the angular spread of different users.
The compression ratio
$\eta$ is selected optimally using
full search to minimize $\eta=\arg\underset{\eta\in(0,1]}{\min}\epsilon$
and this is considered along the simulations section unless mentioned otherwise.
\smallskip
\begin{figure}[h]
\vspace{-0.1cm}
\begin{center}
\vspace{-0.1cm}
\includegraphics[scale=0.65]{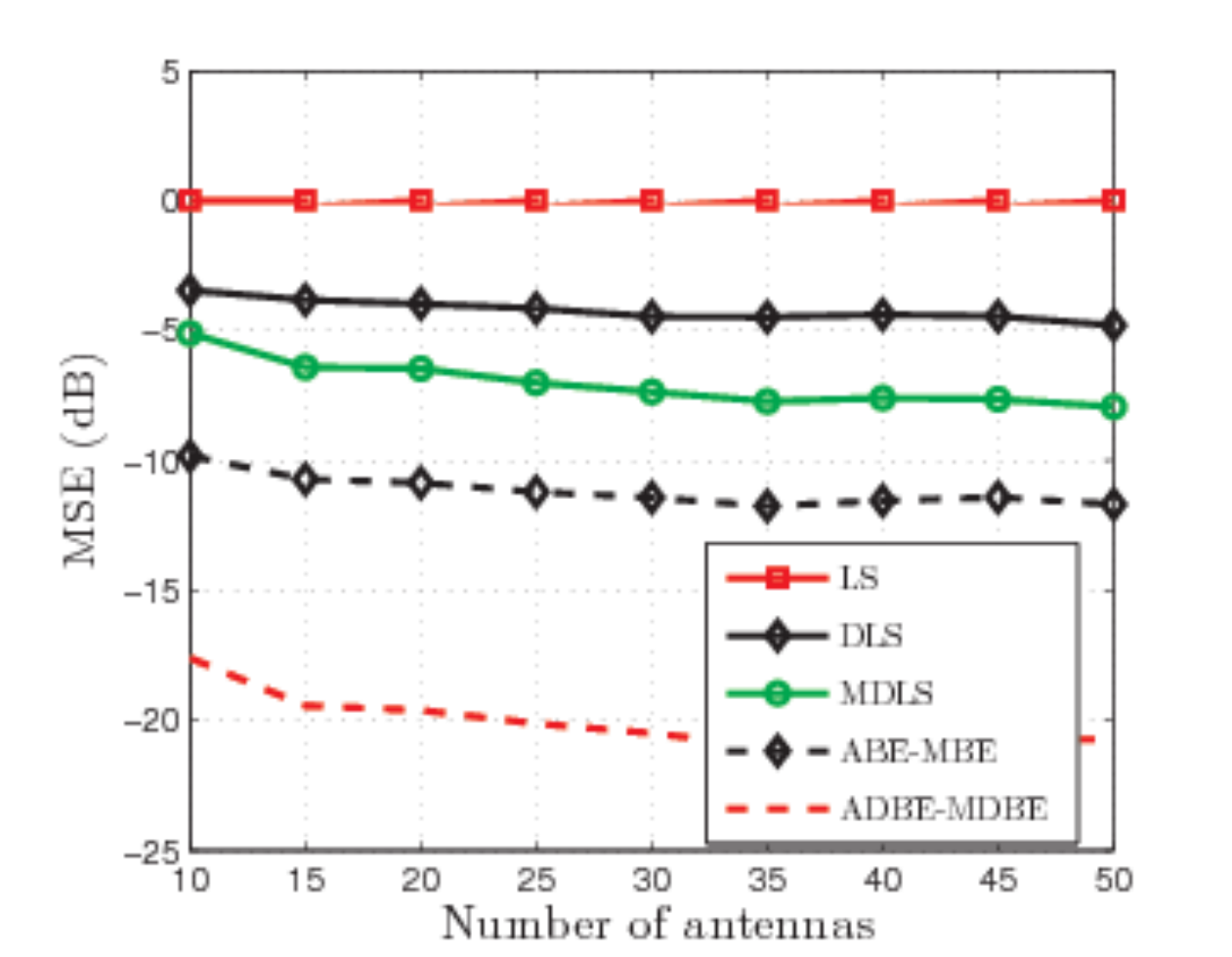}
\vspace{-0.4cm}\caption{\label{ant}\textit{\small The normalized MSE $\epsilon$ vs.
number of antennas, the considered scenario $C=2$, $K=2$, $\beta=1$, $\Delta\theta=20^{\circ}$,
$\theta_{s1}=\{10^{\circ},25^{\circ}\}$, $\theta_{s2}=\{20^{\circ},35^{\circ}\}$,
$\Delta_o\theta=5^{\circ}$, and $P=0$. }}
\end{center}
\end{figure}
\subsection{Accurate Second order statistics}
Fig. (\ref{ant}) illustrates the comparison among the different estimation techniques with respect to the number of antennas. It can be noted 
that MSE performance reduces monotonically with respect to the number of antennas at BS. It can be concluded that the LS has the worst performance
in comparison with the other techniques. This can be explained by the fact
that LS estimation just removes the impact of training sequence without introducing
any processing to the aggregate of the received signals. DLS and MDLS are introduced to exploit the correlation information of the target estimate and incorporate
the concept of DCT to get the important spatial frequency related to each
estimate. It can be inferred that DLS and MDLS  outperform the typical LS
for all antennas scenario. In this figure, we also depict the comparison
between the BE estimation techniques. The proposed techniques DBE and MDBE overcome the
ABE-MBE. Intuitively, the adaptive MBE-BE performs better than
BE and also this applies to DCT based technique, so the figures in this section
display the comparison between the adaptive proposed schemes and the typical
ones.\smallskip

\begin{figure}[h]
\vspace{-0.1cm}
\begin{center}
\vspace{-0.1cm}
\includegraphics[scale=0.65]{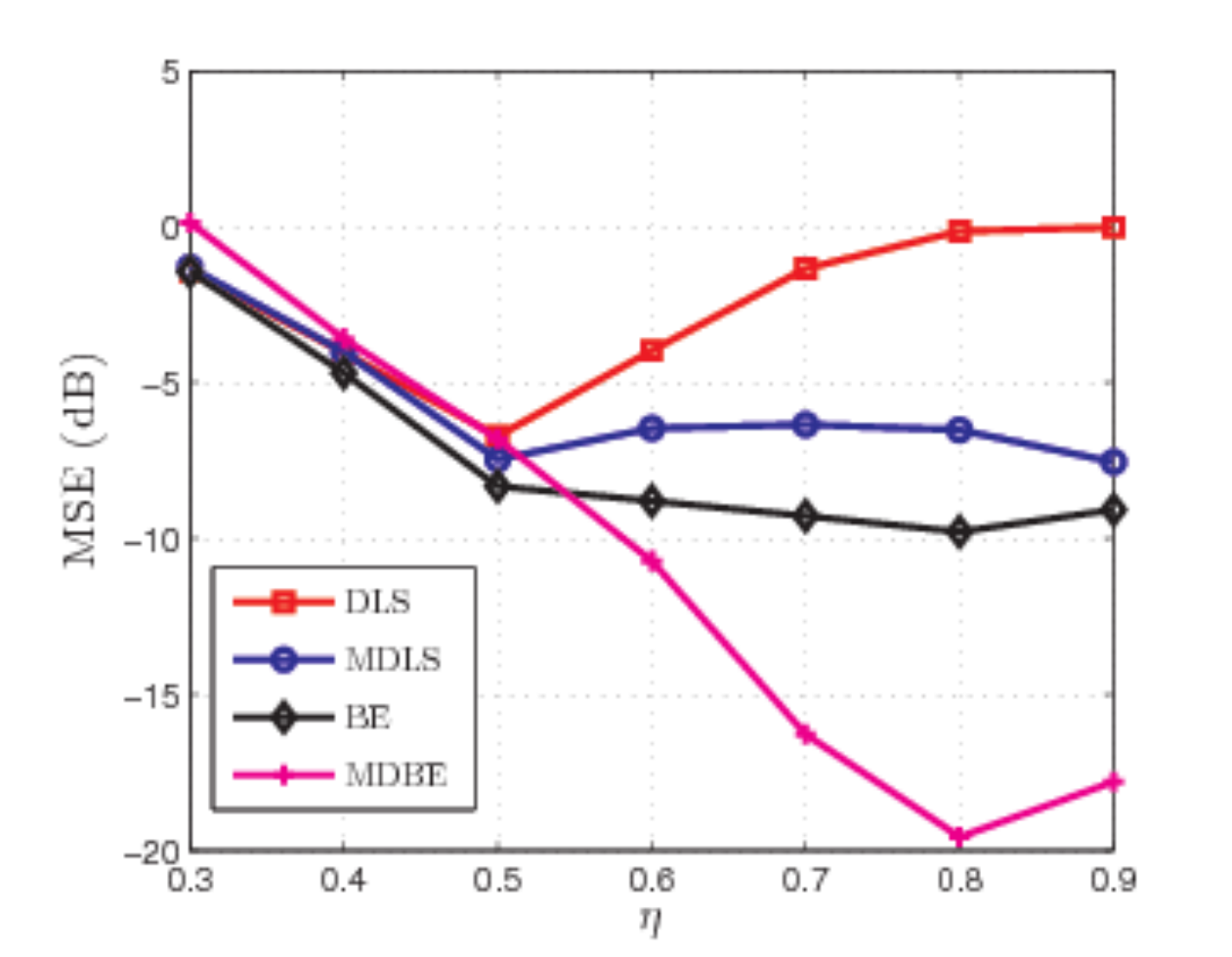}
\vspace{-0.3cm}\caption{\label{eta}\textit{\small The normalized MSE $\epsilon$ vs.
$\eta$, the considered scenario $C=2$, $K=2$, $\beta=1$, $P=0$,  $\theta_{s1}=\{10^{\circ},
30^{\circ}\}$, $\theta_{s2}=\{20^\circ,35^{\circ}\}$ $\Delta\theta=20$ and $\Delta_o\theta=5^{\circ}$.}}
\end{center}
\end{figure}

In Fig. (\ref{eta}), we study the performance of different proposed DCT based
estimation techniques with respect to the compression ratio $\eta$. It can be noted that at high $\eta$, a large percentage of spatial frequencies are utilized. Therefore, the contamination is still contained in the influential
frequencies of each estimate, which makes it hard to get an accurate estimate. On the other hand, when $\eta$ is low, a small percentage of spatial frequencies
are utilized which also removes a part of the useful signal to get an accurate
channel estimate. Fig. (\ref{eta}) shows that for the DCT Based LS estimation
(DLS, MDLS)
the optimal $\eta$ is lower than DCT Based Bayesian estimation
(DBE, MDBE).
This can be justified by the fact that typical BE estimation techniques,
by their construction, have the capability of handling the interference at estimation
in contrast to LS based estimation. This makes the DCT compression step for
DLS and MDLS more valuable as it adds the capability of mitigating the
contamination. While for DBE and MDBE, it enhances the capability of mitigating
the interference, which makes it unnecessary to have lower values of $\eta$.  
\begin{figure}[h]
\begin{center}
\includegraphics[scale=0.65]{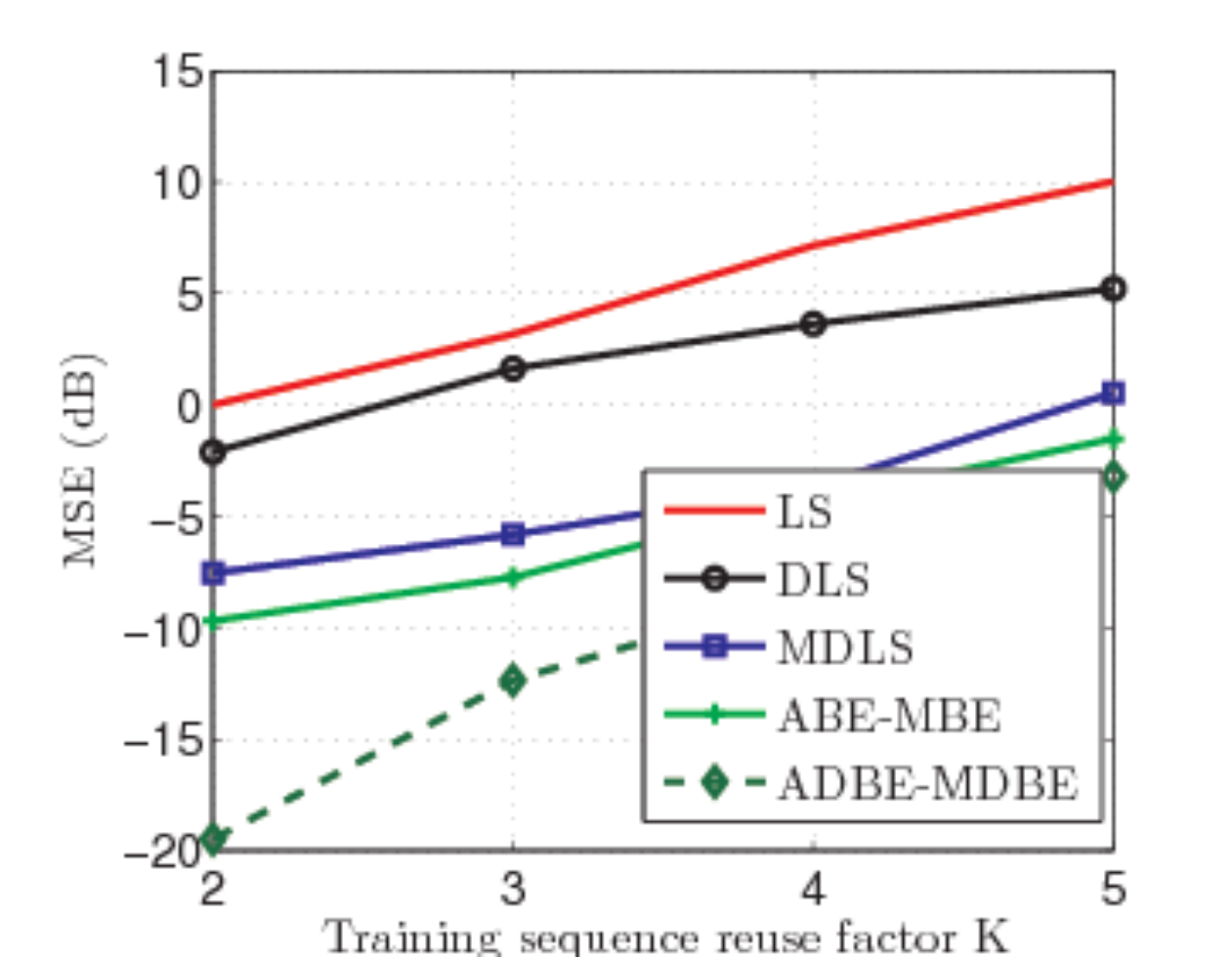}
\caption{\label{K}\textit{\small The normalized MSE vs  training sequence reuse factors $M=20$,$P=0
dB$, $\beta=1$,
$\theta_{s1}=\{10^{\circ},25^{\circ}, 40^{\circ}, 55^{\circ}, 70^{\circ}\}$, $\theta_{s2}=\{20^{\circ},35^{\circ}, 50^{\circ}, 65^{\circ}, 80^{\circ}\}$,
$\theta_{s3}=\{0^{\circ},15^{\circ},30^{\circ},45^{\circ},60^{\circ}\}$,
$\theta_{s4}=\{35^{\circ},50^{\circ},65^{\circ},70^{\circ},85^{\circ}\}$,
$\theta_{s5}=\{85^{\circ},70^{\circ},55^{\circ},40^{\circ},25^{\circ}\}$,
 $\Delta\theta=20$, and $\Delta_o\theta=5^{\circ}$.}}
 \vspace{-0.2cm}
\end{center}
\end{figure}  

Fig. (\ref{K}) illustrates the MSE performance with respect to training sequence reuse factor. It is anticipated that increasing the training sequence reuse
over the network increases the interference levels at the estimation, which
makes it harder to obtain an accurate estimate of the required channel.
This figure plots the comparison of different proposed estimation algorithms.
It can be noted that there is a gap between ADBE-MDBE and
ABE-MBE at low training sequence reuse factors, but this gap reduces with increasing
the training sequence reuse factor. This can hamper the implementation of ADBE-MDBE at
high training sequence reuse factor. However, the simulated scenarios consider 
interference limited case namely $\beta=1$. MDLS performs closely to ABE-MBE for all
training sequence reuse factors. Comparing MDLS with DLS, it can be seen that MDLS
shows  an enhanced performance over DLS, and the both proposed techniques
outperform the typical LS for all training sequence reuse values.\smallskip 

Fig. (\ref{overlap}) displays the comparison of different estimation techniques
with the respect to the angular spread overlap $\Delta_{o}\theta$. The considered scenario. Intuitively, higher overlap of the angular spread leads to
higher training sequence contamination. It can be noticed that DLS and MDLS converge
to LS in the scenario of complete angular spread overlap. While
for BE techniques, it can be viewed that ADLS-MDLS has the same performance
as ABE-MBE in case of complete angular spread overlap. For the scenario
of distinct subspaces $\Delta_o\theta=0^{\circ}$,  the typical BE and MBE and consequently their adaptive schemes perform better than ADBE-MDBE
because there is no contamination at this scenario, and the DCT compression is meaningless.\smallskip
\begin{figure}[h]
\begin{center}
\includegraphics[scale=0.65]{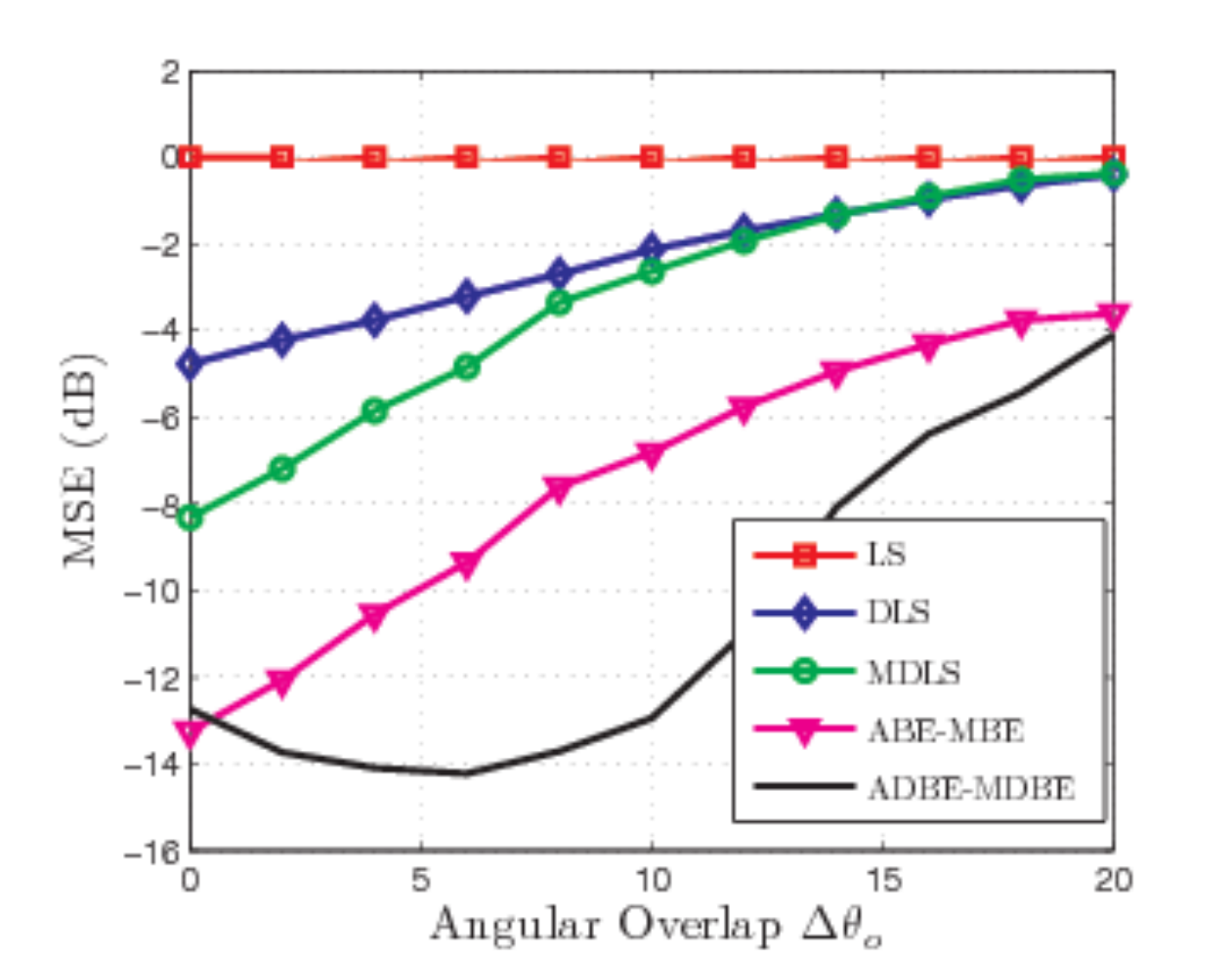}
\vspace{-0.2cm}
\caption{\label{overlap}\textit{\small The normalized MSE vs angular overlap $K=2$, $C=2$,$P=0dB$, $\beta=1$,
$\theta_{s1}=\{10^{\circ},30^{\circ}\}$, $\theta_{s2}=\{70^{\circ},85^{\circ}\}$,
$M=10$.}}
\end{center}
\end{figure}
The
  efficiency of employing training sequence allocation algorithms is depicted in Fig.
(\ref{PA}).
We want to allocate 4 training sequences to 8 users in the two cells. The comparison of each estimation technique without and with
training sequence allocation is plotted and it can be inferred that the system performance
is enhanced if training sequence allocation algorithm is employed. This can be explained
by the fact that selecting the which are assigned the same training sequence enhances the chances of being them naturally separable.
Therefore, the contamination is reduced by employing these algorithms.  
      
It is observed from the simulations that adapting the performance of different BE techniques as ABE-MBE and ADBE-MDBE achieves a better performance than
employing just BE or MBE. Simulations have shown that the frequency of using MBE
(or MDBE) versus
BE (or DBE) in the adaptive algorithm is $75\%$ versus $25\%$ in the scenario of $C=2$,
$K=2$, $\Delta_o\theta=10^{\circ}$ and $\Delta\theta=25^{\circ}$. This observation
shows that the modified algorithm outperforms the original ones in majority
of channel realizations.
\vspace{-0.35cm}
\subsection{The impact of inaccurate second order statistics}
\begin{figure}[h]
\begin{center}
\includegraphics[scale=0.65]{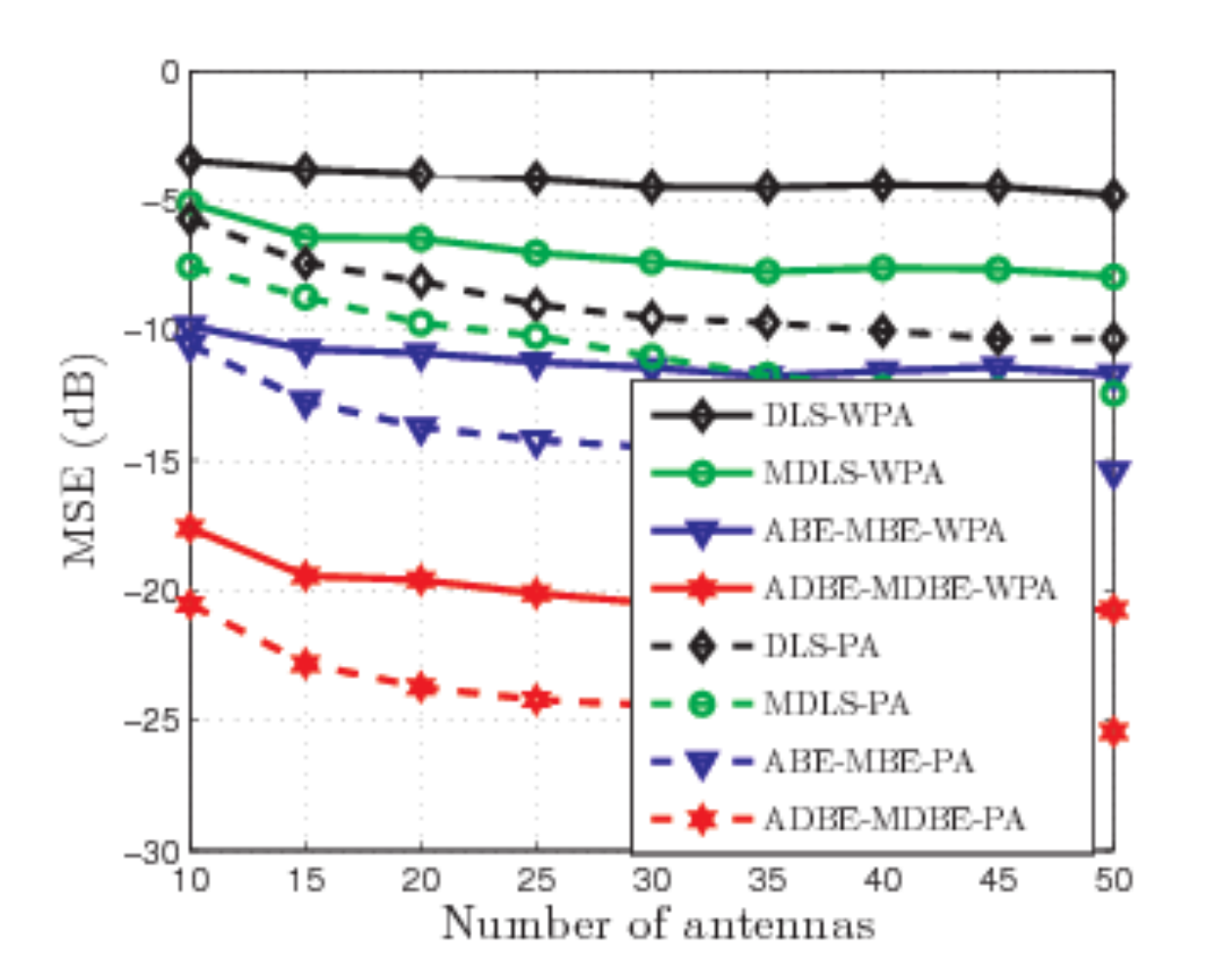}
\vspace{-0.2cm}\caption{\label{PA}\textit{\small The normalized MSE vs number of antennas $M=20$,$P=0
dB$, $\beta=1$,
$\theta_{s1}=\{10^{\circ},40^{\circ}\}$, $\theta_{s2}=\{40^{\circ},60^{\circ}\}$,$\Delta\theta=20$,
and $\Delta_o\theta=5^{\circ}$. }}\end{center}
\end{figure}
\textcolor{black}{
In the previous figures, we study the performance of the proposed techniques
assuming accurate covariance acquisition at all
BSs. Assuming inaccuracies and estimation errors in the covariance acquisition
step may affect the performance of the suggested methods. The impact of these
inaccuracies
is depicted in Fig. (\ref{err}), which plots the relation between the
MSE and the uncertainty. It can be noted that typical LS is not affected
by the uncertainty, this is intuitive since the estimation process does not
depend on the covariance information. This fact also applies on the MSE assuming
DLS estimation. Although the important spatial frequencies determination
depends on the covariance information, the estimator design is independent
from covariance information at BSs. Moreover, it can be concluded that the BE based techniques are more sensitive to covariance errors as they are functions
of the covariance matrices of the involved users,the inaccuracies
affect the contamination rejection at different BSs. These systems' MSE
increase with respect to the amount of covariances inaccuracies. Finally,
since MDLS is modified version of DLS and LS, which is based on increasing the covariance of the target channels, it is expected that the covariance inaccuracies degrade
the estimation performance at different BSs.} \smallskip

\begin{figure}[h]
\begin{center}
\includegraphics[scale=0.65]{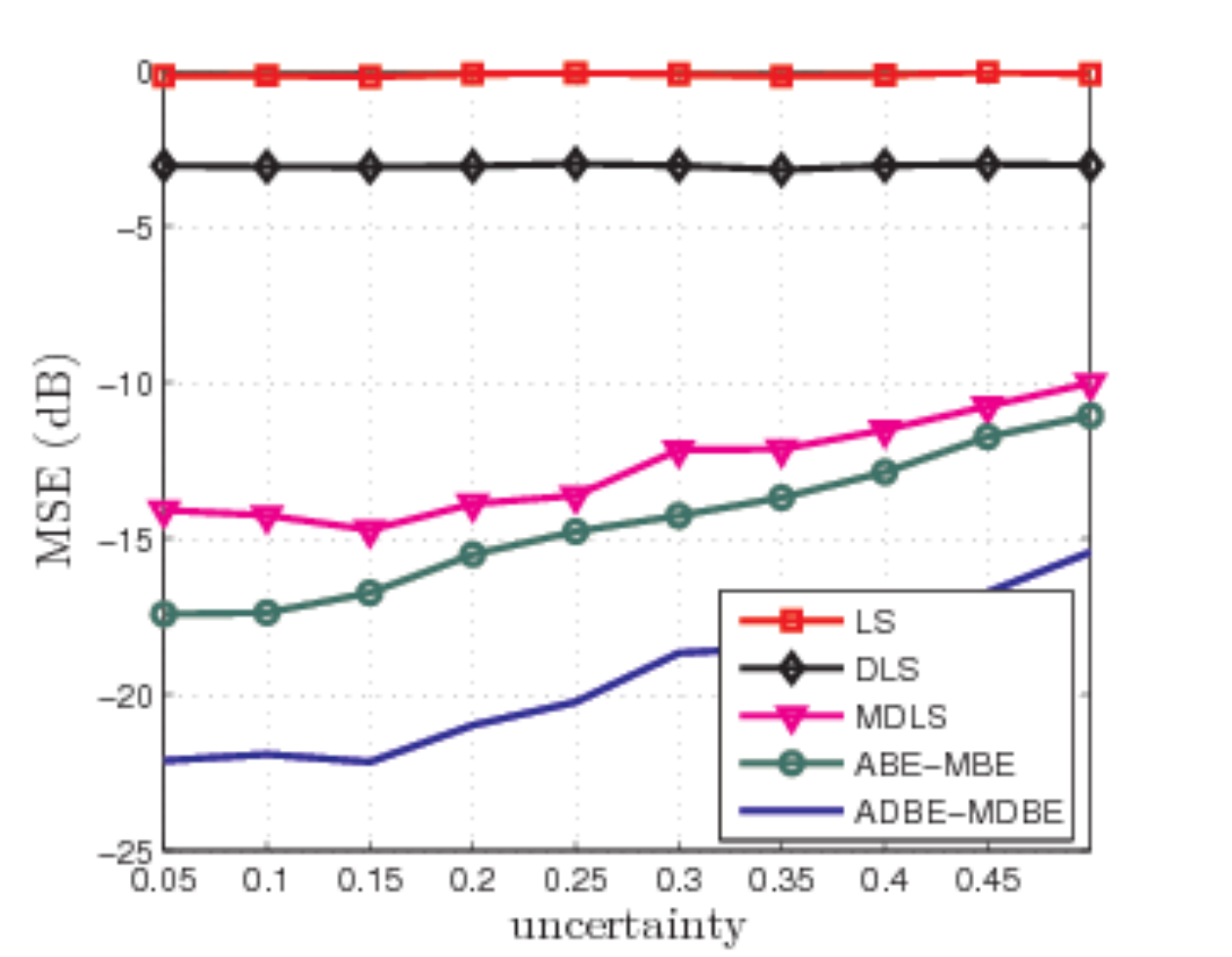}
\vspace{-0.2cm}\caption{\label{err}\textit{\small The normalized MSE vs number of antennas $M=20$,$P=0
dB$, $\beta=1$,
$\theta_{s1}=\{10^{\circ},30^{\circ}\}$, $\theta_{s2}=\{40^{\circ},60^{\circ}\}$,$\Delta\theta=20$,
and $\Delta_o\theta=0^{\circ}$. }}\end{center}
\end{figure}
\textcolor{black}{
The impact of inaccuracies considering the overlap in angular spread 
users' subspaces is depicted in Fig. (\ref{erro}). In comparison with Fig.
(\ref{err}), the performance of the proposed algorithms is studied. It can
be noted that the proposed techniques are more sensitive to uncertainties
which can be translated into higher MSE at all. As conclusion, to protect
the system from these uncertainties, training sequence allocation should
be adapted to take into consideration these uncertainties in their designs.}       
\begin{figure}[h]
\begin{center}
\includegraphics[scale=0.65]{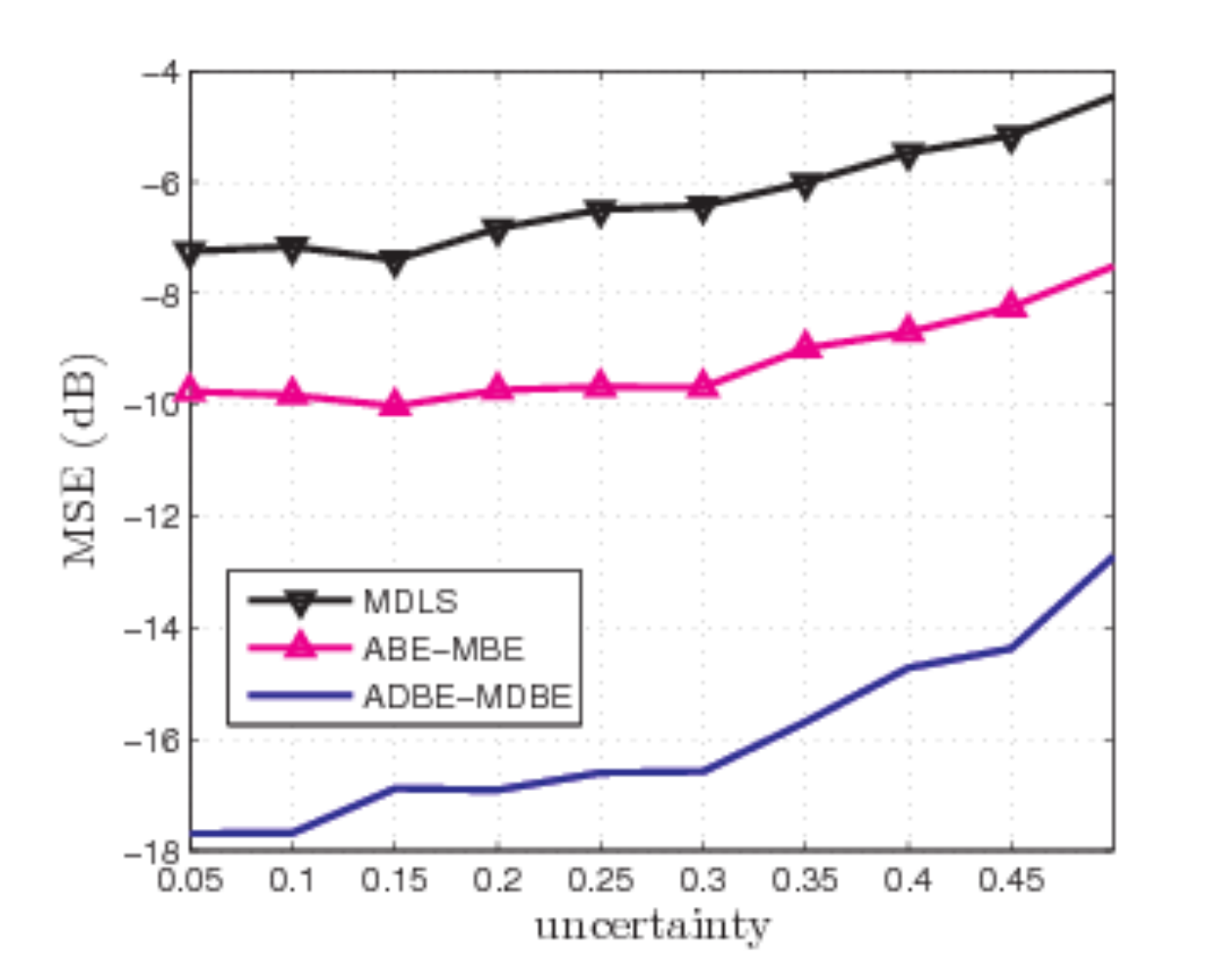}
\vspace{-0.2cm}\caption{\label{erro}\textit{\small The normalized MSE vs number of antennas $M=20$,$P=0
dB$, $\beta=1$,
$\theta_{s1}=\{10^{\circ},40^{\circ}\}$, $\theta_{s2}=\{40^{\circ},60^{\circ}\}$,$\Delta\theta=20$,
and $\Delta_o\theta=5^{\circ}$. }}\end{center}
\end{figure}

\section{Conclusion}
In this work, we studied the interference during the channel estimation  and its impact on the system performance in
multicell multiantenna networks. We investigated the performance of a Bayesian estimation and a least square estimation framework and formulated
the lower bound and the upper bound of mean square error for such estimator. We proposed  modified techniques to enhance the estimation accuracy by introducing the DCT, thereby transforming the problem into a different domain. This allowed
the development of a new interference mitigation algorithm by compressing the spatial frequencies. It enabled enhanced estimation utilizing the DCT compressing
capability to reduce the overlap in the interfering subspaces and boost
the separation in the spatial domain. We incorporated this concept with training
sequence
allocation to assign different training sequences to reduce
 the interference in the overlapping area of the DCT subspaces. The performance of proposed algorithms was studied and compared
to current state of the art techniques. From the simulation results, it can
be concluded that the proposed algorithms provide considerable gains over the
conventional Bayesian and least squares estimation techniques.\smallskip
\vspace{-0.3cm}
\appendix
\vspace{-0.2cm}
\textit{Proof of Lemma 3}: This can be proven by  taking the expectation
of $\mathbb{E}_{\theta}[\mathbf{h}^{d}_{lc}]$ at each spatial frequency $k$  taking into the account the different distribution of the arrival angles
 at the BS. For uniform angular distribution on $[\theta_1,\theta_2]$ 
 \vspace{-0.1cm}
 \begin{eqnarray}
 \nonumber
 \mathbb{E}_{\theta}[\mathbf{h}^d_{lc}[k]]&=&\int^{\theta_2}_{\theta_1}\mathbf{h}^d_{lc}[k]dF_\theta
 =\frac{1}{\theta_2-\theta_1}\int^{\theta_2}_{\theta_1}\mathbf{h}^d_{lc}[k]d\theta
 \end{eqnarray}
Since the term $\omega_i=c\sin(\theta_i)$ exist in all $k$ and make the integration
term non integrable. For $k=0$, the integration can be replaced by another integration,
\begin{eqnarray}\nonumber
\hspace{-0.1cm}\mathbb{E}_{\omega}[\mathbf{h}^d_{lc}[0]]&=&\frac{1}{c(\theta_2-\theta_1)}\int^{\omega_2}_{\omega_1}\mathbf{h}^d_{lc}[0]\sqrt{1-\big(\frac{\omega}{c}\big)^2}d\omega\\
&\leq&\frac{}{}\frac{1}{c(\theta_2-\theta_1)}\int^{\omega_2}_{\omega_1}\mathbf{h}^d_{lc}[0]d\omega
\vspace{-0.1cm}
\end{eqnarray}
 For $k=0$, If the angular spread angles close to zero, the previous term
 has a constant value, while if the angular spread is close to $\frac{\pi}{2}$
 the previous integration is close to zero. For high frequencies,
 using the distributive property of integration, we can solve the related
 to integration by setting $x_1=c\sin\theta-{\frac{k\pi}{M}}$ and $x_2=c\sin\theta+{\frac{k\pi}{M}}$
 as:  \begin{eqnarray}\nonumber
 \hspace{-0.7cm}\mathbb{E}_{\omega}[\mathbf{h}^d_{lc}[k]]&=&\frac{1}{\theta_2-\theta_1}
 \int^{c\sin\theta_2-\frac{k\pi}{M}}_{c\sin\theta_1-\frac{k\pi}{M}}\mathbf{h}^d_{lc}[k]\sqrt{1-(\frac{x_1+\frac{k\pi}{M}}{c})^2}dx_1\\\nonumber
 &+&\frac{1}{\theta_2-\theta_1}\int^{c\sin\theta_2+\frac{k\pi}{M}}_{c\sin\theta_1+\frac{k\pi}{M}}\mathbf{h}^d_{lc}[k]\sqrt{1-(\frac{x_2-\frac{k\pi}{M}}{c})^2}dx_2.
 \end{eqnarray}
For angle of arrival close to zero, this integration will be close to zero.
While for angle of arrival close to $\frac{\pi}{2}$, this term will result
into considerable constant, which proves the lemma.

{\textit{Proof of Lemma 4}: If we assume $\theta_s$ is close to zero, finding the expectation for small angles
close to zero at different spatial frequencies as
\begin{eqnarray}
\vspace{-0.1cm}
c(k)=\lim_{M\rightarrow \infty}\mathbb{E}_{\omega}[|\mathbf{a}[k]|]
\end{eqnarray}  }

It can be found using (\ref{unif})-(\ref{4}) that the zero frequency converges
to constant while the rest frequencies converge to zero. If $\theta_s$ is
close to $90^{\circ}$, all the frequencies will converge to zero, except
for the
highest spatial frequency.

\begin{figure*}[t]
\vspace{-0.5cm}
\small
\begin{tabular}[t]{c}
\begin{minipage}{17 cm}
\begin{eqnarray}
\label{unif}
\small
\hspace{-0.3cm}c(k)=\small\begin{cases}\begin{array}{ccc} \frac{0.5\lambda
j}{2\pi
d\Delta \theta}\Bigg(\log\frac{1-e^{\frac{j2\pi d\theta_2}{\lambda}}}{1-e^{\frac{j2\pi d\theta_1}{\lambda}}}-\log\frac{e^{\frac{-j2\pi
d\theta_2}{\lambda}}(-1+e^{\frac{-j2\pi d\theta_2}{\lambda}})}{e^{\frac{-j2\pi
d\theta_1}{\lambda}}(-1+e^{\frac{-j2\pi d\theta_1}{\lambda}})}\Bigg)+{1},&
&, M=odd,k=0,\\
 \frac{0.5\lambda
j}{2\pi
d\Delta \theta}\Bigg(\log\frac{1-e^{\frac{j2\pi d\theta_2}{\lambda}}}{1-e^{\frac{j2\pi d\theta_1}{\lambda}}}-\log\frac{e^{\frac{-j2\pi
d\theta_2}{\lambda}}(-1+e^{\frac{-j2\pi d\theta_2}{\lambda}})}{e^{\frac{-j2\pi
d\theta_1}{\lambda}}(-1+e^{\frac{-j2\pi d\theta_1}{\lambda}})}\Bigg)& &, M=even,k=0,\\
0 & & k\neq 0.
\end{array}
\end{cases}
\end{eqnarray}\smallskip
\normalsize
\scriptsize
\begin{eqnarray}
\label{B1}
\mathbb{E}_{\omega}[\sqrt{M}|\mathbf{a}^{\text{DCT}}(0)|]=\small\begin{cases}\begin{array}{cccc}\frac{\lambda}{2\pi
d\Delta\theta}\Bigg(\displaystyle\sum^{\frac{M-1}{2}}_{m=1}\frac{\sin
m\theta_2-\sin m\theta_1}{m}+\frac{2\pi
d}{\lambda}(\theta_2-\theta_1)\Bigg)& &,M=odd, &k=0,\\
\frac{\lambda}{2\pi d\Delta\theta}\displaystyle \sum^{\frac{M}{2}}_{m=1}\frac{\sin m\theta_2-\sin m\theta_1}{m}& & ,M=even, &k=0.\\
\\
\end{array}\end{cases}
\end{eqnarray}
\end{minipage}
\vspace{0.1cm}\\
\hline
\hline
\end{tabular}

\end{figure*}
\scriptsize
\begin{figure*}[t]
\hspace{0.2cm}
\begin{tabular}[t]{c}
\begin{minipage}{17 cm}
\begin{eqnarray}
\label{C1}
\sum^{\infty}_{m=0}\frac{\sin m\theta}{m}=\theta+0.5j\Bigg(\log\bigg(1-e^{j\theta}\bigg)
-\log\bigg(e^{-j\theta}(-1+e^{j\theta})\bigg)\Bigg)
\end{eqnarray}
 \begin{eqnarray}
\label{B2}
\hspace{-2.9cm}\small\mathbf{a}^{\text{DCT}}(k)&=&\hspace{-0.1cm}\small\Bigg(\hspace{-0.1cm}{\Bigg(\frac{\sin(\small(\omega_i+\frac{\pi k}{M}\small)\frac{M}{2})}{\sin(\small(\omega_i+\frac{\pi
k}{M}\small
)\frac{1}{2})}\Bigg)^2\hspace{-0.1cm}+\hspace{-0.1cm}\Bigg(\frac{\sin(\small(\omega_i-\frac{\pi k}{M}\small)\frac{M}{2})}{\sin(\small(\omega_i-\frac{\pi
k}{M}\small)\frac{1}{2})}\Bigg)^2
\hspace{-0.1cm}+\hspace{-0.1cm}2\frac{\sin(\small(\omega_i+\frac{\pi k}{M}\small)\frac{M}{2})}{\sin(\small(\omega_i+\frac{\pi
k}{M}\small
)\frac{1}{2})}\frac{\sin(\small(\omega_i-\frac{\pi k}{M}\small)\frac{M}{2})}{\sin(\small(\omega_i-\frac{\pi
k}{M}\small
)\frac{1}{2})}\cos\Big(\small\frac{\pi(M+1)k}{M}\big)\Bigg)^{\frac{1}{2}}}
\end{eqnarray}
\vspace{-0.2cm}
\begin{eqnarray}
\label{4}
\mathbf{a}^{\text{DCT}}(k)=\begin{cases}\begin{array}{ccc}
\sqrt{{2}} \sqrt{\Big(\frac{\sin(\small(\omega_i+\frac{\pi k}{M}\small)\frac{M}{2})}{\sin(\small(\omega_i+\frac{\pi
k}{M}\small
)\frac{1}{2})}-\frac{\sin(\small(\omega_i-\frac{\pi k}{M}\small)\frac{M}{2})}{\sin(\small(\omega_i-\frac{\pi
k}{M}\small
)\frac{1}{2})}\Big)^2}& \approx 0& ,M={odd}, k\neq 0,\\
\sqrt{2}\sqrt{\Big(\frac{\sin(\small(\omega_i+\frac{\pi k}{M}\small)\frac{M}{2})}{\sin(\small(\omega_i+\frac{\pi
k}{M}\small
)\frac{1}{2})}+\frac{\sin(\small(\omega_i-\frac{\pi k}{M}\small)\frac{M}{2})}{\sin(\small(\omega_i-\frac{\pi
k}{M}\small
)\frac{1}{2})}\Big)^2}& \approx 0&,M={even}, k\neq 0.
\end{array}\end{cases}
\end{eqnarray}
\end{minipage}
\vspace{0.2cm}\\
\hline
\hline
\end{tabular}
\vspace{-0.5cm}
\end{figure*}

\end{document}